\RequirePackage{amsmath,amssymb,amsthm}

\documentclass[12pt]{iopart}

\usepackage{amsfonts}
\usepackage{graphicx}
\usepackage{geometry}
\usepackage{braket}
\usepackage{multicol}
\usepackage{subfig}
\usepackage{wrapfig}
\usepackage{caption}
\usepackage{setspace}
\usepackage[english]{babel}
\usepackage[colorlinks]{hyperref}
\usepackage{titlesec}

\begin{document}

\title[]{Spekkens' toy model in all dimensions and its relationship with stabilizer quantum mechanics}

\author{Lorenzo Catani\textsuperscript{1} and Dan E. Browne\textsuperscript{1}}

\address{\textsuperscript{1}University College London, Physics and Astronomy department, Gower St, London WC1E 6BT, UK}
\ead{lorenzo.catani.14@ucl.ac.uk}

\begin{abstract}
Spekkens' toy model is a non-contextual hidden variable model with an epistemic restriction, a constraint on what an observer can know about reality. The aim of the model, developed for continuous and discrete prime degrees of freedom, is to advocate the epistemic view of quantum theory, where quantum states are states of incomplete knowledge about a deeper underlying reality. Many aspects of quantum mechanics and protocols from quantum information can be reproduced in the model. 

In spite of its significance, a number of aspects of Spekkens' model remained incomplete. Formal rules for the update of states after measurement  had not been written down, and the theory had only been constructed for prime-dimensional, and infinite dimensional systems. In this work, we remedy this, by deriving measurement update rules, and extending the framework to derive models in all dimensions, both prime and non-prime.

Stabilizer quantum mechanics is a sub-theory of quantum mechanics with restricted states, transformations and measurements. First derived for the purpose of constructing error correcting codes, it now plays a role in many areas of quantum information theory. Previously, it had been shown that Spekkens' model was operationally equivalent in the case of infinite and odd prime dimensions. Here, exploiting known results on Wigner functions, we extend this to show that Spekkens' model is equivalent to stabilizer quantum mechanics in all odd dimensions, prime and non-prime. This equivalence provides new technical tools for the study of technically difficult compound-dimensional stabilizer quantum mechanics.

\end{abstract}

\maketitle

\section{Introduction}

A long tradition of research, starting from the famous ``EPR paper'' \cite{EPR}, has consisted of analysing quantum theory in terms of hidden variable models, with the aim of obtaining a more intuitive understanding of it. This has led to some crucial results in foundation of quantum mechanics, namely Bell's and Kochen-Specker's no-go theorems \cite{Bell}\cite{Kochen}. Nowadays a big question is whether to interpret the quantum state according to the ontic view, \emph{i.e.} where it completely describes reality, or to the epistemic view, where it is a state of incomplete knowledge of a deeper underlying reality which can be described by the hidden variables.
In 2005, Robert Spekkens \cite{Spek1} constructed a non-contextual hidden variable model to support the epistemic view of quantum mechanics. The aim of the model was to replace quantum mechanics by a hidden variable theory with the addition of an epistemic restriction (\emph{i.e.} a restriction on what an observer can know about reality). The first version of the model \cite{Spek1} was developed in analogy with quantum bits (qubits), with 2-outcome observables. Despite the simplicity of the model, it was able to support many phenomena and protocols that were believed to be intrinsically quantum mechanical (such as dense coding and teleportation). Spekkens' toy model has influenced much research over the years: \emph{e.g.} people provided a new notation for it \cite{Pusey}, studied it from the categorical point of view \cite{Coecke}, used it for quantum protocols \cite{Damian}, exploited similar ideas to find a classical model of one qubit \cite{Pawel}, and tried to extend it in a contextual framework \cite{Larsson}. Also Spekkens' toy model addresses many key issues in quantum foundations: whether the quantum state describes reality or not, finding a derivation of quantum theory from intuitive physical principles and classifying the inherent non-classical features.

A later version of the model \cite{Spek2}, which we will call Spekkens' Theory (ST), introduced a more general and mathematically rigorous formulation,  extending the theory to systems of discrete prime  dimension, where \textit{dimension} refers to the maximum number of distinguishable measurement outcomes of observables in the theory, and continuous variable systems.
Spekkens called these classical statistical theories with epistemic restrictions as \emph{epistricted statistical theories}. By considering a particular epistemic restriction that refers to the symplectic structure of the underlying classical theory, the \textit{classical complementarity principle}, theories with a rich structure can be derived. Many features of quantum mechanics are reproduced there, such as Heisenberg uncertainty principle, and many protocols introduced in the context of quantum information, such as teleportation. However, as an intrinsically non-contextual theory, it cannot reproduce quantum contextuality (and the related Bell non-locality), which, therefore arises as the signature of quantumness. Indeed, for odd prime dimensions and for continuous variables, ST was shown to be operationally equivalent to sub-theories of quantum mechanics, which Spekkens called quadrature quantum mechanics.

In the finite dimensional case quadrature quantum mechanics is better known as \emph{stabilizer quantum mechanics} (SQM). The latter is a sub-theory of quantum mechanics developed for the description and study of quantum error correcting codes \cite{GottesPhD} but subsequently playing a prominent role in many important quantum protocols. In particular, many studies of quantum contextuality can be expressed in the framework of SQM, including the GHZ paradox \cite{GHZ} and the Peres-Mermin square \cite{Mermin}\cite{Peres}. This exposes a striking difference between odd and even dimensional SQM. Even-dimensional SQM contains classical examples of quantum contextuality while odd-dimensional SQM exhibits no contextuality at all, necessary for its equivalence with Spekkens' Theory. While developed for qubits, SQM was rapidly generalised to systems of arbitrary dimension,  \cite{GottesPhD}. However, for non-prime dimensions SQM remains poorly characterised and little studied (recent progress in this was recently reported in \cite{Gottes2}).

In spite of its importance, there remain some important aspects of Spekkens' Theory which have not yet been characterised and studied. First of all, all prior work on ST have only considered systems where the dimension is prime. Furthermore, while Spekkens' recent work strengthens the mathematical foundations of the model \cite{Spek2}, one key part of the theory has not yet been described in a general and rigorous way. These are the measurement update rules, the rules which tell us how to update a state after a measurement has been made. In prior work, these rules, and the principles behind them have been described but not formalised. 

In this paper, we complete this step, deriving a formal description of the measurement rules for prime-dimensional ST. Having done so, we now have a fully formal description of the model, which can be used as a basis to generalise it. We do so, generalising the framework from prime-dimensions to arbitrary dimensions and finding that it is the measurement update rule, where the richer properties of the non-prime dimension can be seen, which provides the key to this generalisation.

Having developed ST for all finite dimensions, we then focus on the general odd-dimensional case, and prove that in all odd-dimensional cases Spekken's Theory is equivalent to Stabilizer Quantum Mechanics. The bridge between SQM and ST is given by Gross' theory (GT) of discrete Wigner function \cite{Gross}. Unlike most other studies, Gross' treatment considered both prime and non-prime cases in its original formulation.

To summarise the contributions of this paper,  we provide a compete formulation of ST in \emph{all} discrete dimensions, even and odd, endowed with the updating rules for sharp measurements both for prime and non-prime dimensional systems. We extend the equivalence between ST and SQM via Gross' Wigner functions to \emph{all} odd dimensions,  and  find the measurement updating rules also for the Wigner functions. 
The above equivalence allows us to shed light onto a complete characterisation of SQM in non-prime dimensions.
Finally the incredibly elegant analogy between the three theories in odd dimensions: ST, SQM and GT, is depicted in terms of their updating rules.

The remainder of the paper is structured as follows.
In section 1 we precisely and concisely describe the original framework of Spekkens' theory, in particular we define ontic and epistemic states, observables and the rule to obtain the outcome of the measurement of an observable given a state.
In section 2 and 3 we state and prove the updating rules in Spekkens' theory respectively for prime and non-prime dimensional systems. We prove these in two steps: first considering the case in which the state and measurement commute, and then the more general (non-commuting) case.  The mathematical difference between the set of integers modulo \emph{d}, for \emph{d} prime and non-prime, results in having two levels of observables: the fundamental ones - the fine graining observables - and the ones that encode some degeneracy - the coarse-graining observables. The latter are problematic and are only present in the non-prime case. This is the reason why we need a different formulation in the two cases. The updating rules for the coarse graining observables will need a step in which the coarse-graining observables are written in terms of fine graining ones.
In section 4 we state the equivalence of ST and SQM via Gross' Wigner functions in all odd dimensions. We also express the already found updating rules in terms of Wigner functions and we use them to depict the elegant analogies between these three theories.  
The paper ends with a discussion of the possible applications of our achievements and with a summary of the main results.

\newpage \section{Spekkens' theory}
\label{SecSpek}

We start by reviewing and introducing Spekkens' theory for prime-dimensional systems. We take a slightly different approach to \cite{Spek1} and \cite{Spek2}. ST is a hidden variable theory, where the hidden variables are points in a phase space. The state of the hidden variables is called the \textit{ontic state}. In Spekkens' model the ontic state is hidden and can never be known by an experimenter. The experimenter's best description of the system is the \textit{epistemic state}, representing a probability distribution over the points in phase space.

For a single $d$-dimensional system, a phase space can be defined via the values of two conjugate fiducial variables, which we label $X$ and $P$, in analogy to position and momentum. $X$ and $P$ can each take any value between $0$ and $d-1$, and a single ontic state of the system is specified by a pair $(x,p)$, where $x$ is the value of $X$ and $p$ is the value of $P$. This phase space is equivalent to the space $\mathbb{Z}_d^{2}$. In figure \ref{SystemEx}  three examples of epistemic states of one trit ($d=3$) are depicted, where $X$ and $P$ are represented by the rows and columns in the phase space $\mathbb{Z}_3.$

\begin{figure}[h!]
\centering
\subfloat[][\label{a}]
{\includegraphics[width=.45\textwidth,height=.18\textheight]{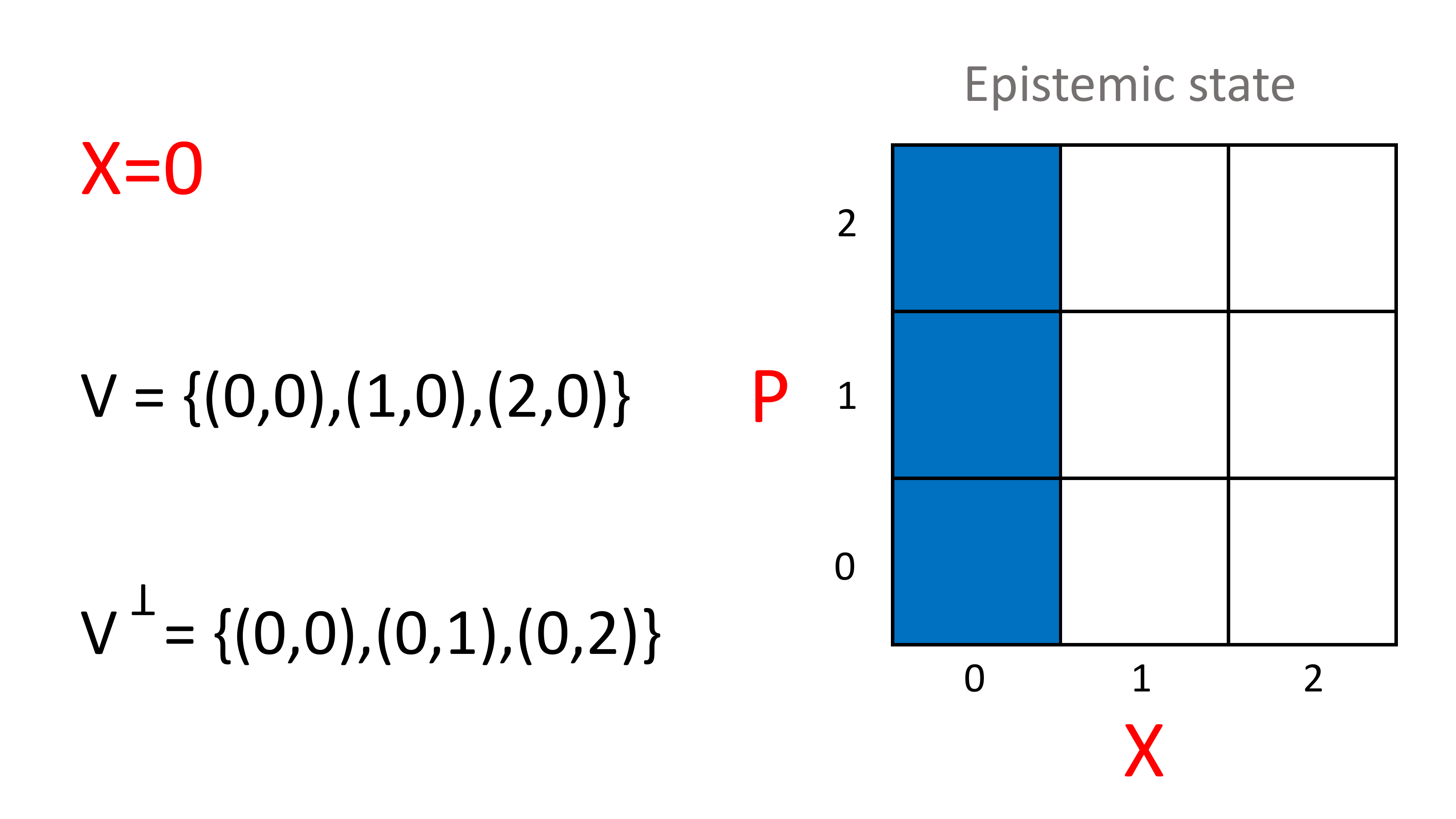}}\hfil
\subfloat[][\label{b}]
{\includegraphics[width=.45\textwidth,height=.18\textheight]{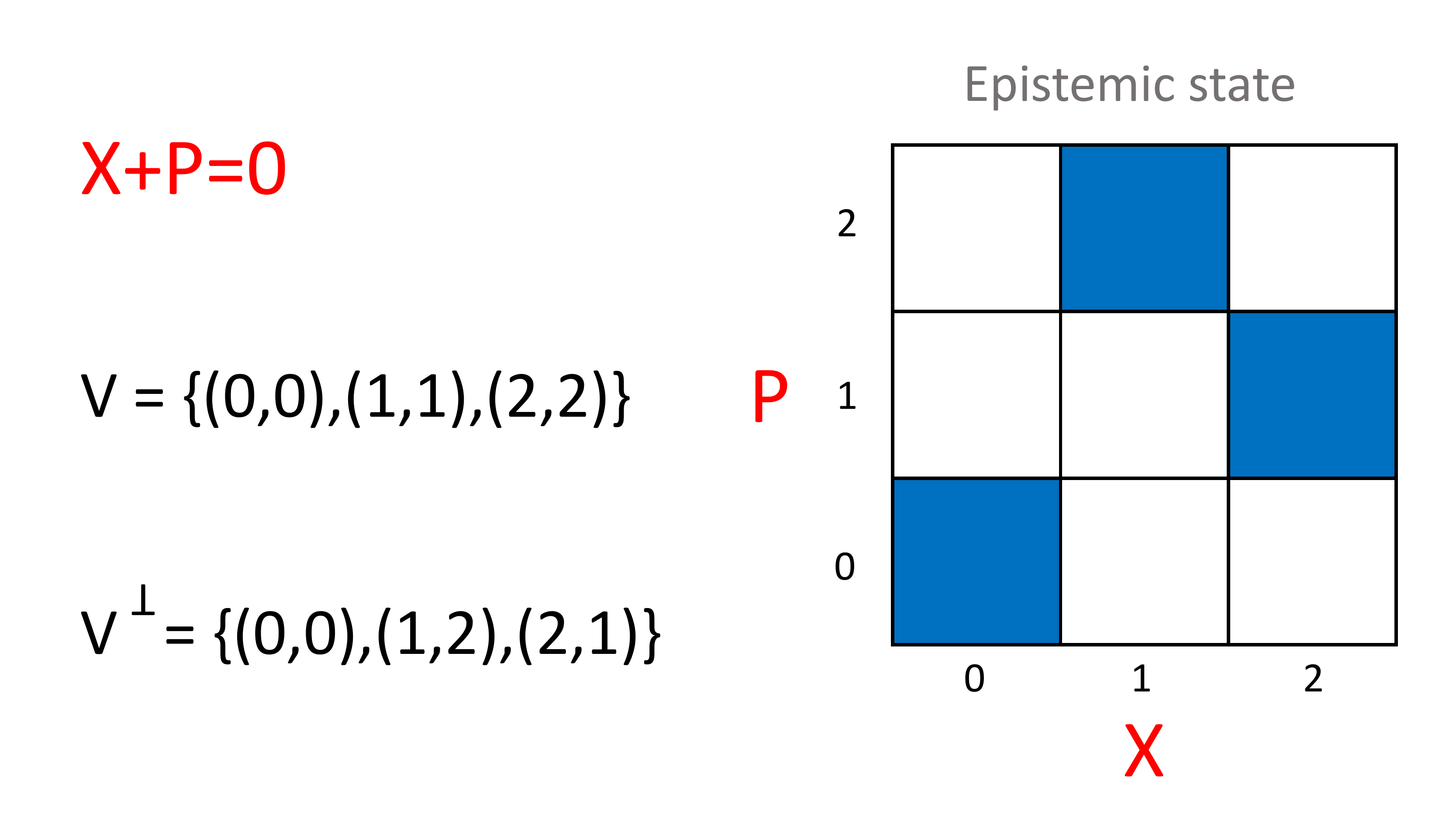}}\hfil
\subfloat[][\label{c}]
{\includegraphics[width=.45\textwidth,height=.18\textheight]{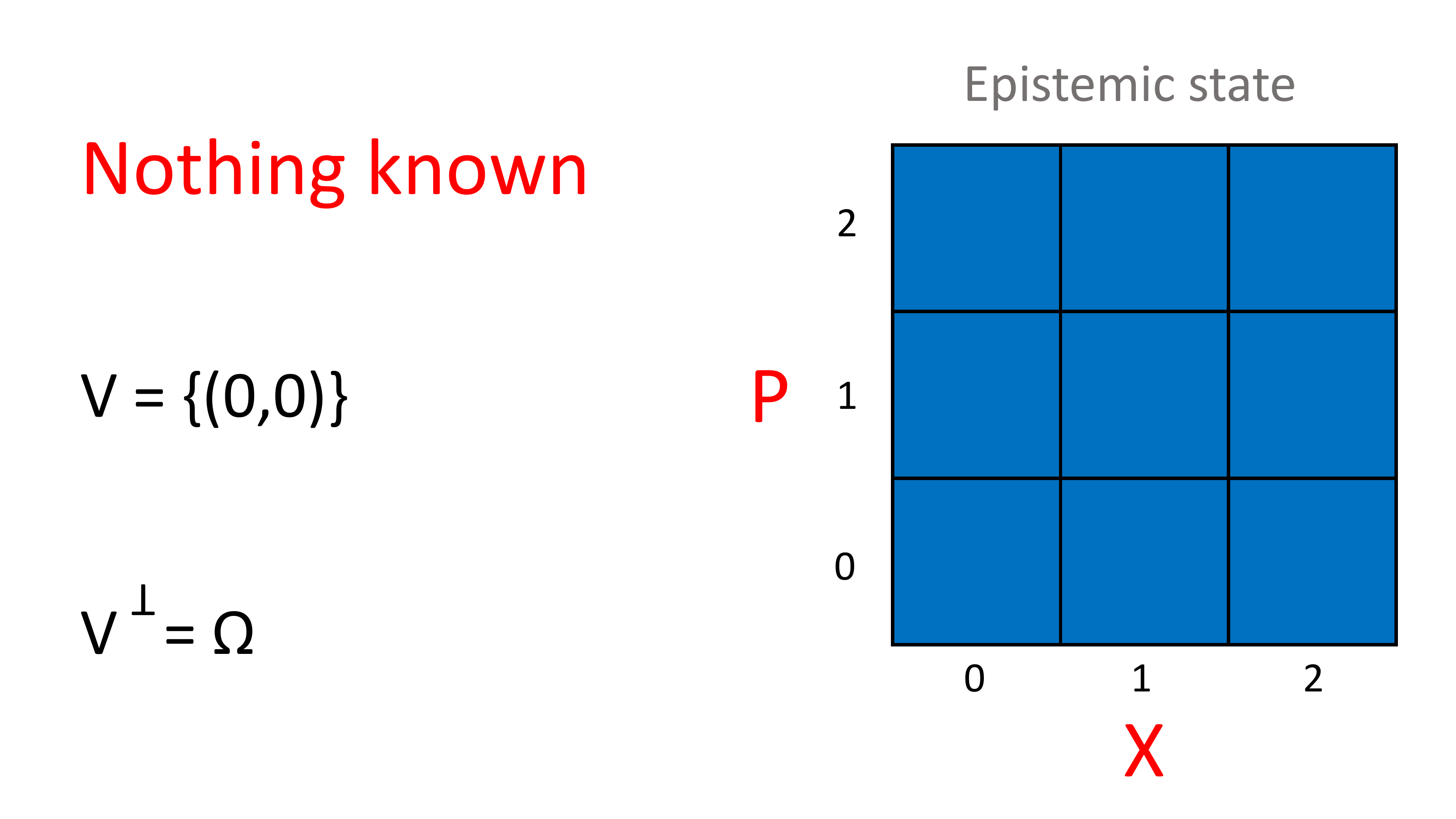}}

\caption{\footnotesize{\textbf{One trit Spekkens states examples.} In the figures above we consider the case of one trit and we find the isotropic subspaces $V$ and $V^{\perp}$ and the corresponding Spekkens epistemic state. In these cases the observables(linear functionals) are always of the form $aX+bP=0,$ where $a,b\in \mathbb{Z}_3.$ Moreover in the above examples we assume $\bold{w}=\bold{0}.$ In figure \ref{a} the observer only knows $X=0$ and this implies that the generator of $V$ is $\bold{\Sigma}=(1,0).$ The subspace $V^{\perp}$ can be simply calculated from $V$ by definition. In figure \ref{b} the observer only knows that $X+P=0$ and this implies the generator of $V$ to be $\bold{\Sigma}=(1,1).$ In figure \ref{c} nothing is known. The subspace $V$ is generated by $\bold{\Sigma}=(0,0)$ only. Here $V^{\perp}$ coincides with the whole phase space $\Omega.$
Note that it is not possible to have $V^{\perp}=(0,0),$ because this would correspond to have the knowledge of the ontic state.}}

\label{SystemEx}
\end{figure}

A collection of $n$ systems is described by $n$ pairs of independent conjugate variables $X_j$ and $P_j$, with $j\in{0,\dots,n-1}$ a label indexing the systems. The phase space, denoted by $\Omega$, is simply the cartesian product of single system phases spaces and thus $\Omega\equiv(\mathbb{Z}_d)^{2n}.$\footnote{The dimension $d$ is any positive number, and we will not, in general, restrict it to odd or even, prime or non-prime, unless specified.} 

The ontic state of the $n$-party system represents a set of values for each fiducial observables $X_j$ and $P_j$. In other words, an ontic state is denoted by a point in the phase space $\lambda\in\Omega.$ We call  $X_j$ and $P_j$ observables because they correspond to  measurable quantities, and assume that these observables are sufficient to uniquely define the ontic state. %\footnote{A crucial aspect of Spekkens' model is that the states are defined through measurements (\emph{e.g.} in the bit case this consists of pointing out how many binary questions (i.e. measurements) can be answered about where the ontic state is).}
We can refer to $\Omega$ as a vector space where the ontic states are vectors (bold characters) whose components (small letters) are the values of the fiducial variables: \begin{equation}\bold{\lambda}=(x_0,p_0,x_1,p_1,\dots,x_{n-1},p_{n-1}).\end{equation}

Not only are the fiducial variables important for defining the state space, they also generate the set of all general observables in the theory. A generic observable, denoted by $\Sigma,$ is defined by any linear combination of fiducial variables: \begin{equation}\Sigma=\sum_m(a_mX_m+b_mP_m),\end{equation} where $a_m,b_m\in \mathbb{Z}_d$ and $m\in{0,\dots, n-1}.$ %Note that we are ignoring the identity observable (thus dropping a constant term) because it would complicate the notation.
The observables inhabit the dual space $\Omega^{*},$ which is isomorphic to $\Omega$ itself. % \footnote{To be more precise, this is the set of linear functionals. By definition, the set of linear functionals $\Omega^*$ of a given space $\Omega$ is called the dual space of $\Omega.$ Since we are in the finite case, $\Omega$ and $\Omega^*$ have the same dimension and $\Omega^*=\Omega$. This last comment explains why we can use the space of linear functionals (and, in particular, their vector representation) to describe our phase space.}  
Therefore we can define them as vectors, in analogy with ontic states, \begin{equation}\bold{\Sigma}=(a_0,b_0,a_1,b_1,\dots,a_{n-1},b_{n-1}).\end{equation}
The formalism provides a simple way of \emph{evaluating} the outcome $\sigma$ of any observable measurement $\Sigma$ given the ontic state $\lambda,$ \emph{i.e.} by computing their \emph{inner product}: \begin{equation}\label{linear}\sigma=\Sigma^T\lambda=\sum_j(a_jx_j+b_jp_j),\end{equation} where all the arithmetic is over $\mathbb{Z}_d.$

Spekkens' theory gains its special properties, and in particular, its close analogy with stabilizer quantum mechanics via the imposition of an \textit{epistemic restriction}, a restriction on what an observer can know about the ontic state of a system. The observer's best description is called the \emph{epistemic state}, which is represented by a probability distribution $p(\lambda)$ over $\Omega$ (figure \ref{SystemEx}).

The epistemic restriction of ST is called \emph{classical complementarity principle} and it states that two observables can be simultaneously measured only when their Poisson bracket is zero. This is motivated by Stabilizer Quantum Mechanics, since it captures the condition for two observables in SQM to \textit{commute}. We shall adopt the quantum terminology here, and say that if the Poisson bracket between two observables is zero they commute. This can be simply recast in terms of the \emph{symplectic inner product}:%\footnote{Note that the symplectic geometry comes in only with the introduction of the joint measurability principle for observables.} 
\begin{equation}\left\langle\bold{\Sigma_1},\bold{\Sigma_2}\right\rangle\equiv\bold{\Sigma_1}^TJ\bold{\Sigma_2}=0,\end{equation} where $J= \bigoplus_{j=1}^{n} \begin{bmatrix} 0 & 1 \\ -1 & 0 \end{bmatrix}_j$ is the symplectic matrix.
Note that each observable $\Sigma_j$ partitions $\Omega$ into $d$ subsets, each of the form $(span\{\Sigma_j\})^{\perp}+w,$ where $w$ is any ontic state such that $\bold{\Sigma_j}^T\cdot\bold{w}=\sigma_j.$

Let us now consider sets of variables that can be jointly known by the observer. Such variables commute, and represent a sub-space of $\Omega$ known as an \emph{isotropic subspace}. We denote the subspace of the known variables as $V=span\{\Sigma_1,\dots,\Sigma_n\}\subseteq\Omega,$ where $\Sigma_i$ denotes one of the generators (commuting observables) of $V$.

Sets of known commuting variables are important as these define the epistemic states within the theory. In particular, we can define an epistemic state by the set of  variables $V$ that are known by the observer and also the values ${\sigma_1,\dots,\sigma_n}$ that these variables take.

This means that $\bold{\Sigma}_j^T\cdot \bold{w}=\sigma_j,$ where $w\in V$ is an ontic state that evaluates the known observables. We will call $w$ a \emph{representative ontic state} for the epistemic state. More precisely we can state the following theorem.
\newtheorem{EpistemicTheorem}{Proposition}
\begin{EpistemicTheorem}\label{EpistemicTheorem}
The set of ontic states consistent with the epistemic state described by $(V,\bold{w})$ is \begin{equation}V^{\perp}+\bold{w},\end{equation} where the perpendicular complement of $V$ is, by definition, $V^{\perp}=\{a\in\Omega \; |\; \bold{a}^T \bold{b}=0\; \forall\; b\in V \}.$
\end{EpistemicTheorem}

\begin{proof}
Let us start by considering the set of ontic states $\lambda$ such that $\bold{\Sigma}_j^T\bold{\lambda}=0 \; \forall j.$ By definition of perpendicular complements, the ontic states $\lambda$ belong to $V^{\perp}.$ If we consider an ontic state $w$ such that $\bold{\Sigma}_j^T \bold{w}=\sigma_j,$ then $\bold{\Sigma}_j^T (\bold{\lambda}+\bold{w})=\sigma_j.$ Therefore the ontic states consistent with the epistemic state associated to $(V,\bold{w})$ are the ones of the kind $\lambda+w,$ \emph{i.e.} the ones belonging to $V^{\perp}+\bold{w}.$ 
\end{proof}
Note that the presence of $\bold{w}\neq \bold{0}$ simply implies a translation, that is why we can also call it \emph{shift vector}.
%In particular we can see the need for this translation as follows: the set $V^{\perp}+\bold{w}$ is not a subspace because $\{(0,0)\}$ does not belong to the set. If we shift back every point of this set by $\bold{w}$ we get the subspace $V$. Therefore we can think of the shift vector as deriving from the need to have a subspace.

By assumption the probability distribution associated to the epistemic state $(V,\bold{w})$ is uniform (indeed we expect all possible ontic states to be equiprobable), so the probability distribution of one of the possible ontic states in the epistemic state $(V,\bold{w})$ is \begin{equation}\label{distribution}P_{(V,\bold{w})}(\bold{\lambda})=\frac{1}{d^n}\delta_{V^{\perp}+\bold{w}}(\lambda),\end{equation} where the delta is equal to one only if $\bold{\lambda}\in V^{\perp}+\bold{w}$ (note this means that the theory is a \emph{possibilistic} theory). 
In figure \ref{SystemEx} we specify the subspaces $V$ and $V^{\perp}$ in three different examples of epistemic states of one trit.

We can sum up our approach to Spekkens' model as follows:
\begin{enumerate}
\item Start from the intuitive (physically justified) formula \eqref{linear} that relates observables $\Sigma_j$, ontic states $\lambda$ and outcomes $\sigma_j$.
\item Epistemic restriction: the compatible observables are the ones whose symplectic inner product is zero.
\item Compute the shift vector $\bold{w}.$ This allows us to shift back the set of points $\bold{\lambda}$ to obtain a subspace.
\item The set of ontic states compatible with the epistemic state $(V,\bold{w})$ is $V^{\perp}+\bold{w},$ where $V$ is the isotropic subspace spanned by the observables $\Sigma_j$ (the set of known variables).
\end{enumerate}
We say that this approach is physically intuitive because we start with equation \eqref{linear}, which is physically motivated and states, observables and the corresponding outcomes are defined in terms of it. Equation \eqref{linear} also allows us to see that the shift comes from the need to recover the subspace structure.

\section{Updating rules - \emph{prime} dimensional case} 
The formulation of ST in \cite{Spek2}, made for prime (and infinite) dimensional systems and described in the previous section, does not provide a full treatment of the transformative aspect of measurements, \emph{i.e.} how the epistemic state has to be updated after a measurement procedure. In the following we will provide a proper formalization of it, and in the next section we will generalise the formalism to all dimensions, non-prime too.

The set of integers modulo $d$ shows different features depending on $d$ being prime or not. In particular in the non-prime case it is not always possible to uniquely define the inverse of a number. The consequences of this will directly affect the updating rules. In particular the possible observables sometimes will not show full spectrum: some outcomes will not be possible because they would derive from arithmetics involving numbers with not well-defined inverses. This will divide the set of possible observables in two categories depending on whether they have full spectrum or not. We start from the prime case where problematic observables are not present because inverses always exist.

Like in quantum theory, duality in the description of states and measurements characterises ST. This means that we can represent the elements of a measurement $\Pi$ in an epistemic-state way, $(V_{\Pi},\bold{r}),$ where we can go from one element of the measurement to the other by simply shifting the representative ontic vector $\bold{r}$ (see figure \ref{meas}). In ST the measurement process corresponds to the process of \emph{learning} some information (\emph{aka} asking questions) about the ontic state of the system. According to the classical complementarity principle only the observables that are compatible (\emph{i.e.} Poisson-commute) with the state of the system can be learned (jointly knowable). This means that the state after measurement will be given by the generators of the state before the measurement and the generators of the measurement which are compatible with it.\footnote{As an abuse of language we here talk of generators of a state meaning the orthogonal basis set that generates the subspace of known variables associated with the state.}
It is then fundamental to understand how compatible sets of ontic states (the isotropic subspaces of known variables $V$ and their perpendicular $V^{\perp}$) change when independent observables are added and removed from the set of known variables $V$.

\begin{figure}[h!]
\centering
%\subfloat[][\label{a}]
{\includegraphics[width=.45\textwidth,height=.18\textheight]{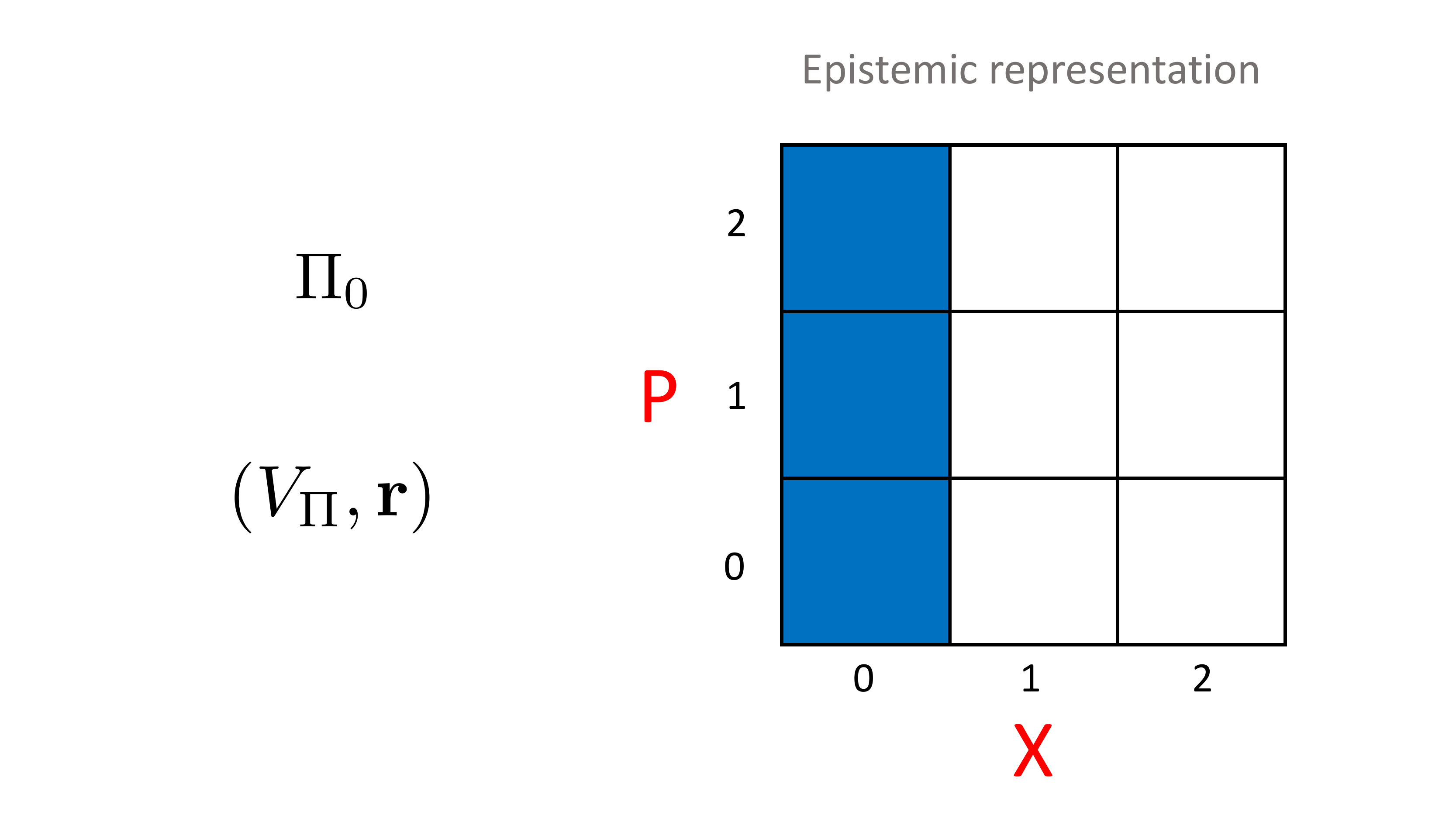}}\hfil
%\subfloat[][\label{b}]
{\includegraphics[width=.45\textwidth,height=.18\textheight]{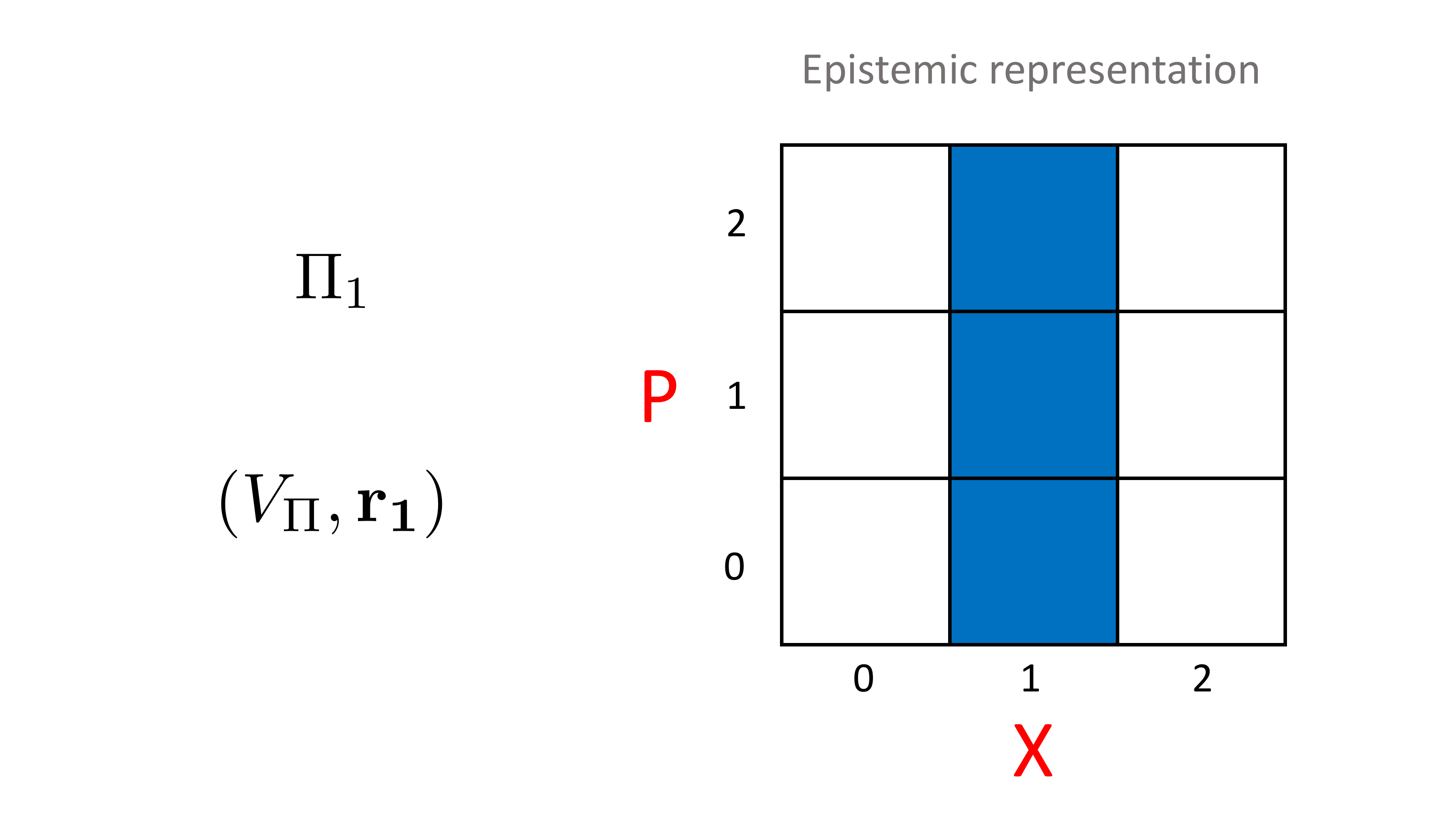}}\hfil
%\subfloat[][\label{c}]
{\includegraphics[width=.45\textwidth,height=.18\textheight]{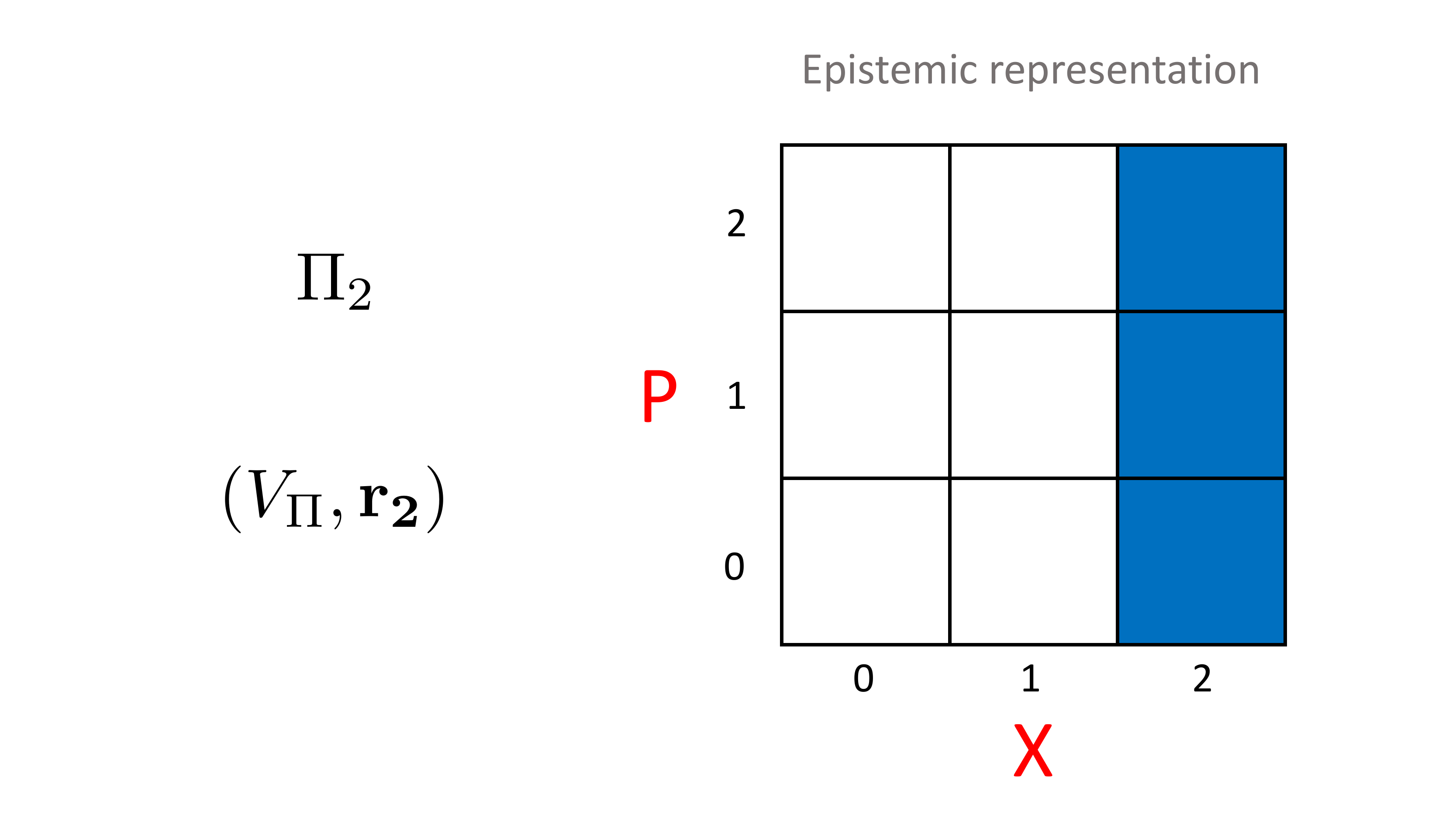}}

\caption{\footnotesize{\textbf{Epistemic representation of a measurement.} The elements of the measurement $\Pi$ can be represented as epistemic states. This duality is present also in quantum theory. The elements of the measurement $\Pi_0,\Pi_1,\Pi_2$ can be thought as the analogue of the projectors $\{\ket{0}\bra{0}, \ket{1}\bra{1}, \ket{2}\bra{2}\}.$  We can always go from one element to the other by shifting the representative ontic vector. In the above case we can go, for example, from $\Pi_0$ to $\Pi_1$ by adding to $\bold{r}=(0,0)$ the vector $(1,0),$ thus getting $\bold{r}_1=(1,0).$ The example above shows the measurement corresponding to asking the question "what is the value of the variable X?" about the ontic state of the system. }}

\label{meas}
\end{figure}

\subsection{Adding and removing generators to/from V}\label{add}

\begin{enumerate}

\item Let us start with the case of \emph{adding} a generator $\Sigma'$ to the set of generators of $V=span\{\Sigma_1,\dots,\Sigma_n\}.$  We assume that $\Sigma'$ is linear independent with respect to the set spanned by the $\Sigma_j.$ Let us see what happens to $V^{\perp}.$
The subspace $V$ after the addition becomes \begin{equation}\label{adding}V'=V\oplus span\{\Sigma'\}.\end{equation}
By definition the direct sum of two subspaces $A\oplus B$ returns a subspace such that for each $a\in A$ and $b\in B$, the sum $a+b$ belongs to $A\oplus B.$ The direct sum of two subspaces is a subspace.
We are interested in the orthogonal complement of a direct sum. It is well known that $(A\oplus B)^{\perp}=A^{\perp}\cap B^{\perp}.$
This means that by adding a generator to $V,$ its perpendicular $V^{\perp}$ is given by \begin{equation}\label{intersec}V'^{\perp}=V^{\perp}\cap (span\{\Sigma'\})^{\perp}.\end{equation} Note that $V'^{\perp}$ is smaller than $V^{\perp}.$

\item We now analyse what happens if we \emph{remove} a generator, say $\Sigma_n,$ from the set of generators of $V.$ This means that now $V'=span\{\Sigma_1,\dots, \Sigma_{n-1}\}.$  The set $V^{\perp}$ is clearly contained in $V'^{\perp},$ since any vector orthogonal to all elements of $V$ must also be orthogonal to all elements of $V'.$
By definition, the set $V'^{\perp}$ is composed by all the ontic states $\lambda$ such that $\bold{\Sigma}_j^T \bold{\lambda}=0$ for all $j<n,$ but $\bold{\Sigma}_n^T \bold{\lambda}\neq 0.$
This means that we need to remove the constraint $\bold{\Sigma}_n^T \bold{\lambda} = 0$ to enlarge $V^{\perp}$ to $V'^{\perp},$ \emph{i.e.} 
%In order to do that we want to add ontic states $\lambda'$ to $V^{\perp}$ without affecting the generators $\Sigma_j$ of $V.$
%This is possible if the ontic states $\lambda'$ are such that $\Sigma_n^T\bold{\lambda}'=\sigma_n.$ 
 we simply need to add the ontic states $\lambda'=c\gamma$ to $V^{\perp},$ where $c\in\mathbb{Z}_d\neq0$ and $\gamma$ is a vector such that $\Sigma_n^T\gamma=1.$ 
Indeed this implies that \[\Sigma_n^T(\bold{\lambda}+\bold{\lambda}')=\Sigma_n^T(\bold{\lambda}+c\gamma)=0+c\neq 0.\] 
In \emph{prime} dimensions $\gamma$ uniquely exists and it corresponds to $k^{-1}\Sigma_n,$ where $k=\Sigma_n^T\Sigma_n.$ Indeed the inverse of an integer $k\in\mathbb{Z}_d\neq 0$ always uniquely exists if $d$ is a prime number. %In the above expression we are assuming $\Sigma_n^T\Sigma_n=1,$ \emph{i.e.} that $\Sigma_n$ has an inverse. Note that, in general, this is true only for systems of \emph{prime} dimensions. %Moreover we could more generally have $\Sigma_n^T\Sigma_n=k,$ where $k$ is an integer such that $k\neq0.$
The formula for $V'^{\perp}$ then reads
\begin{equation} V'^{\perp}=\bigcup_c (V^{\perp}+ ck^{-1}\Sigma_n)\equiv \bigcup_{w_n\in V_n}(V^{\perp}+w_n) = V^{\perp}\oplus V_n,\end{equation} where the addition of $+ w_n$ means that the whole set $V^{\perp}$ is shifted by $w_n,$ and $V_n=span\{\Sigma_n\}.$ 
The previous trick in general works as follows. Given the ontic state $\lambda,$ the observable $\Sigma$ and the outcome $\sigma$ associated with them, \emph{i.e.} $\bold{\Sigma}^T\bold{\lambda}=\sigma,$ then it is possible to shift the value $\sigma$ by a constant $k$ such that $\Sigma_n^T\Sigma_n=k,$ by only adding $\Sigma$ itself to the ontic state: \begin{equation}\label{identity}\Sigma^T(\bold{\lambda}+\Sigma)=\sigma +\Sigma^T\Sigma=\sigma + k.\end{equation}
Note that the above identity allows us to change the value of the outcome associated with an ontic state by a constant factor (that we can also choose) without affecting any commuting observable (in this case $\Sigma$).

\end{enumerate}

\subsection{Measurement updating rules} \label{uprules}

We now want to find the updating rules for the state $(V,\bold{w})$ of a prime dimensional system when we perform a measurement $(V_{\Pi},\bold{r})$ on it. We will consider $V_{\Pi}$ being spanned by the generators denoted as $\Sigma'_j.$ The representative ontic vector associated to the measurement, $\bold{r},$ is such that, by definition, $\bold{\Sigma'}_j^T\bold{r}=\sigma'_j,$ where the $\sigma'_j$ are the outcomes associated with the measurement.   
The subspace of known variables $V$ can be written in terms of the sets generated by the generators Poisson-commuting with all the $\Sigma'_j,$ $V_{commute},$ and non-commuting ones, $V_{other}.$ According to this definition $V_{commute}$ will always be a subspace. We cannot state the same for $V_{other},$ since the null vector does not belong to it. For this reason we augment $V_{other}$ with the null vector in order to create a subspace.
This implies that we can decompose $V$ as \begin{equation}V=V_{commute}\oplus V_{other}.\end{equation}
We can also prove the following lemma.
\newtheorem{LemmaNonCommuting}{Lemma}
\begin{LemmaNonCommuting}\label{LemmaNonCommuting}
The subspace $V_{other}$ has dimension $m,$ where $m$ is the number of non-commuting generators of the measurement with the state.
\end{LemmaNonCommuting}

\begin{proof}
Let us initially assume the measurement to consist only of one non-commuting generator $\Sigma',$ so $m=1.$ Let us prove the lemma by contradiction. Let $u,v$ be two orthogonal non-zero elements of $V_{other}.$ Note that, by definition of a subspace, if $u,v\in V_{other},$ also a linear combination of $u,v$ has to belong to $V_{other}.$ 
By definition $u,v$ do not commute with $\Sigma'.$ Therefore we can write \[\bold{\Sigma}'^T J \bold{u}=a,\] \[\bold{\Sigma}'^T J \bold{v}=b,\] where $a,b\neq0.$ In particular there will exist a constant $c\in \mathbb{Z}_d$ such that $a-bc=0.$ This implies that \[\bold{\Sigma}'^T J (\bold{u}-c\bold{v})=0.\] Hence the linear combination $(u-cv)$ belongs to $V_{commute}.$ This is a contradiction, therefore $V_{other}$ has dimension $1.$
From the same reasoning, in the case of $m$ non-commuting generators of the measurement, the subspace $V_{other}$ has dimensions at maximum equal to $m.$ Let us assume now that the dimension of $V_{other}$ is $m-1.$ This is not possible because it would mean that, for example, $\Sigma'_{m-1}$ can be written as a linear combination of $\Sigma'_0,\dots,\Sigma'_{m-2}.$ However this is not the case because, by definition of basis set, all the generators are linearly independent. Therefore $V_{other}$ has dimension $m.$
\end{proof}

%It is important to point out that we can find examples where more than one generator of the state do not commute with the measurement $\Sigma',$ but the subspace $V_{other}$ is anyway a one-dimensional one (this means the non-commuting generators are not independent). The algorithm to find the generator of $V_{other}$ is simply the search for the first non-commuting generator among all the generators of $V.$

%We are now ready to state the updating rules for Spekkens theory. Just remember that we are considering a starting epistemic state $(V,\bold{w}),$ defined by the probability distribution \eqref{distribution}, where $V=span\{\Sigma_1,\dots,\Sigma_n\}$ is the isotropic subspace of known variables and $\bold{w}$ is the representative ontic vector that evaluates the known variables, \emph{i.e.} $\bold{\Sigma}_j\bold{w}=\sigma_j,$ where $\sigma_j$ is the outcome associated with the observable generator $\Sigma_j$ of $V.$ Let us assume that we perform the measurement $\Sigma'$ on the above state. More precisely $\Sigma'$ is the generator (we here suppose it is the only one) that spans the isotropic subspace associated to the measurement $V_{\Pi}=span\{\Sigma'\}.$ The representative ontic vector associated to the measurement is $\bold{r},$ such that, by definition, $\bold{\Sigma'}^T\bold{r}=\sigma',$ where $\sigma'$ is the outcome associated with the measurement. 
We will now provide the updating rules both for $V$ and $\bold{w}$ in two steps: first considering the state and measurement to commute, and then the general (non-commuting) case. 

%%%%%%%%%%%%%%%%%%%%%%%%%%%%%

\newtheorem{Theorem}{Theorem}
\begin{Theorem}\label{CommutingTheorem}
\emph{Commuting case}. The epistemic state $(V,\bold{w})$ after a measurement $(V_{\Pi},\bold{r})$ that commutes with it, \emph{i.e.} their generators all Poisson commute, is described by the epistemic state $(V',\bold{w'})$ such that \begin{equation}\label{ultimate}V'^{\perp}=(V^{\perp}+\bold{w}-\bold{w'})\cap(V_{\Pi}^{\perp}+\bold{r}-\bold{w'}),\end{equation} where $\bold{w'}$ is given by equation \begin{equation}\label{represontic}\bold{w'} = \bold{w}+\sum_{i}\bold{\Sigma'}_i^T(\bold{r}-\bold{w})\bold{\gamma}_i,\end{equation} where $\bold{\Sigma'}_i$ are the generators of the measurement $\Pi$ and $\gamma_i$ is such that $\bold{\Sigma'}_i^T\bold{\gamma}_i=1.$
\end{Theorem}

\begin{proof}

When the state and measurement commute we have to add the generators of the measurement to the set of generators of $V,$ as we have seen in the previous subsection \ref{add} (learning stage). Therefore the updating rule for the subspace $V$ is (equation \eqref{adding}) \begin{equation}\label{subcomm} V\rightarrow V'=V\oplus span\{\Sigma'_0,\Sigma'_1,\dots \Sigma'_i,\dots\}= V\oplus V_{\Pi}.\end{equation} In terms of perpendicular subspaces this implies that $V'^{\perp}=V^{\perp}\cap V_{\Pi}^{\perp}.$

%The representative ontic vector $\bold{w}.$ Note that $\bold{w}$ also needs to be updated because now it can also belong to $span\{\Sigma'\}.$
Let us initially assume the measurement to consist only of one generator $\Sigma'.$ 
Let us recall that the outcome associated with $\Sigma'$ is $\sigma'.$ We assume $\bold{w}$ is not compatible with this outcome, \emph{i.e.} $\bold{\Sigma'}^T\bold{w}=\sigma'+x,$ for some shift $x\in \mathbb{Z}_d,$ and we want to find $\bold{w}'$ such that \begin{equation}\bold{\Sigma'}^T\bold{w'}=\sigma'.\end{equation}
The identity \eqref{identity} we used in the previous section does the job. More precisely, \[\bold{w}'=\bold{w}-x\bold{\gamma},\] where the vector $\bold{\gamma}$ is such that $\Sigma'^{T}\gamma=1.$ The above expression can be also written as \[\bold{w}'= \bold{w}-k^{-1}x\bold{\Sigma}',\] where $k=\Sigma'^{T}\Sigma'.$ The inverse of $k$ always exists because we are in the prime dimensional case. Without referring to $x$ we can restate the updating rule for the representative ontic vector as \begin{equation}\label{represonticproof}\bold{w}\rightarrow \bold{w}+k^{-1}(\sigma'-\bold{\Sigma'}^T\bold{w})\bold{\Sigma'}= \bold{w}+k^{-1}\bold{\Sigma'}^T(\bold{r}-\bold{w})\bold{\Sigma'}.\end{equation} Note that if we consider more than one generator of the measurement, we simply have to sum over all those generators in the second term. This immediately follows from considering the whole measurement $\Pi$ as a sequence of measurements given by each generator $\Sigma'_i$ and apply every time the rule \eqref{represonticproof}.
We state again that the above formula always holds for prime dimensional systems. We cannot claim the same in non-prime dimensions.
The correct updating rule for the subspace $V'^{\perp}$ is found by combining the updating rules for $V$ and $\bold{w}$ as in \eqref{ultimate}.
This correction simply sets the subspaces to the same origin in order to correctly compute their intersection, as schematically shown in figure \ref{Venn}. At the end we obtain for the epistemic state $(V',\bold{w'})$ that $V'^{\perp}+\bold{w'}=(V^{\perp}+\bold{w})\cap(V_{\Pi}^{\perp}+\bold{r}).$
We recall that the probability associated to each ontic state consistent with the epistemic state is uniform, $i.e.$ given by \[P(V',\bold{w'}) =\frac{1}{|V'^{\perp}+\bold{w'}|}=\frac{1}{|V'^{\perp}|} = \frac{1}{|(V^{\perp}+\bold{w})\cap(V_{\Pi}^{\perp}+\bold{r})|},\] where $|\cdot|$ indicates the size of the subspace.

\end{proof}

Figure \ref{Commuting_example} shows a basic example of theorem \ref{CommutingTheorem}.

\begin{figure*}[h!]
\centering

{\includegraphics[width=.8\textwidth,height=.35\textheight]{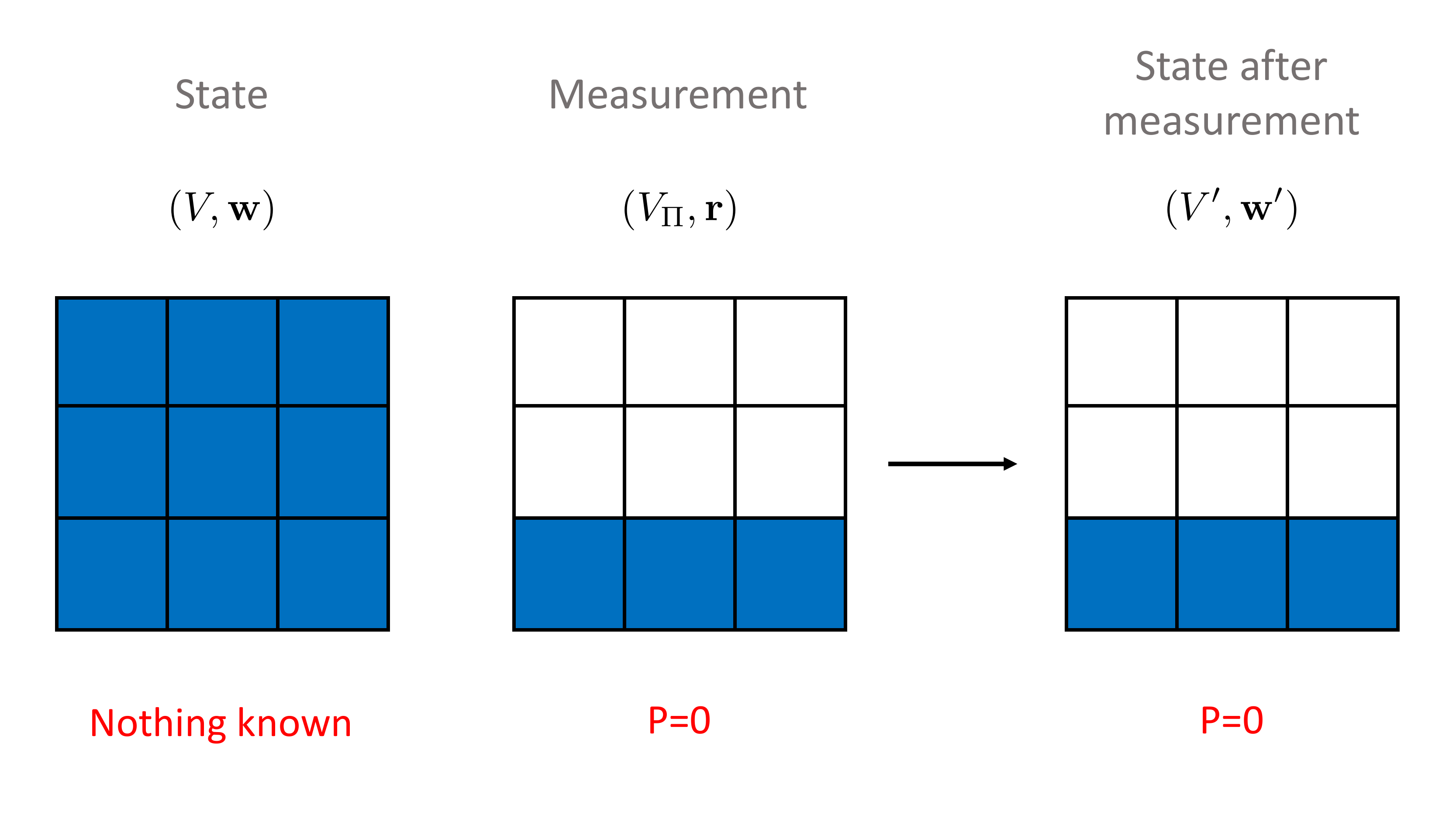}}

\caption{\footnotesize{\textbf{Updating rules in the prime commuting case.} The figure above shows a simple one-trit example of theorem \ref{CommutingTheorem} regarding the updating rule to predict the state after a sharp measurement that commutes with the original state. The state after measurement is given by $V'^{\perp}+\bold{w'}=(V^{\perp}+\bold{w})\cap(V_{\Pi}^{\perp}+\bold{r}).$ In the above case the shift vectors are all $(0,0),$ the perpendicular subspaces are $V^{\perp}=\Omega,$ $V^{\perp}_{\Pi}=span\{(1,0)\},$ and $V'^{\perp}=V^{\perp}_{\Pi}.$ Note that with "measurement" we are here representing one element of the measurement. The other elements can be obtained by simply shifting $\bold{r}$ as seen in figure \ref{meas}. The final state is associated to each element of the measurement, each one with a corresponding probability of happening. The same reasoning holds for figures \ref{NonCommuting_example} and \ref{NonPrime_example}.}}

\label{Commuting_example}
\end{figure*}

\begin{figure*}[h!]
\centering

{\includegraphics[width=.99\textwidth,height=.4\textheight]{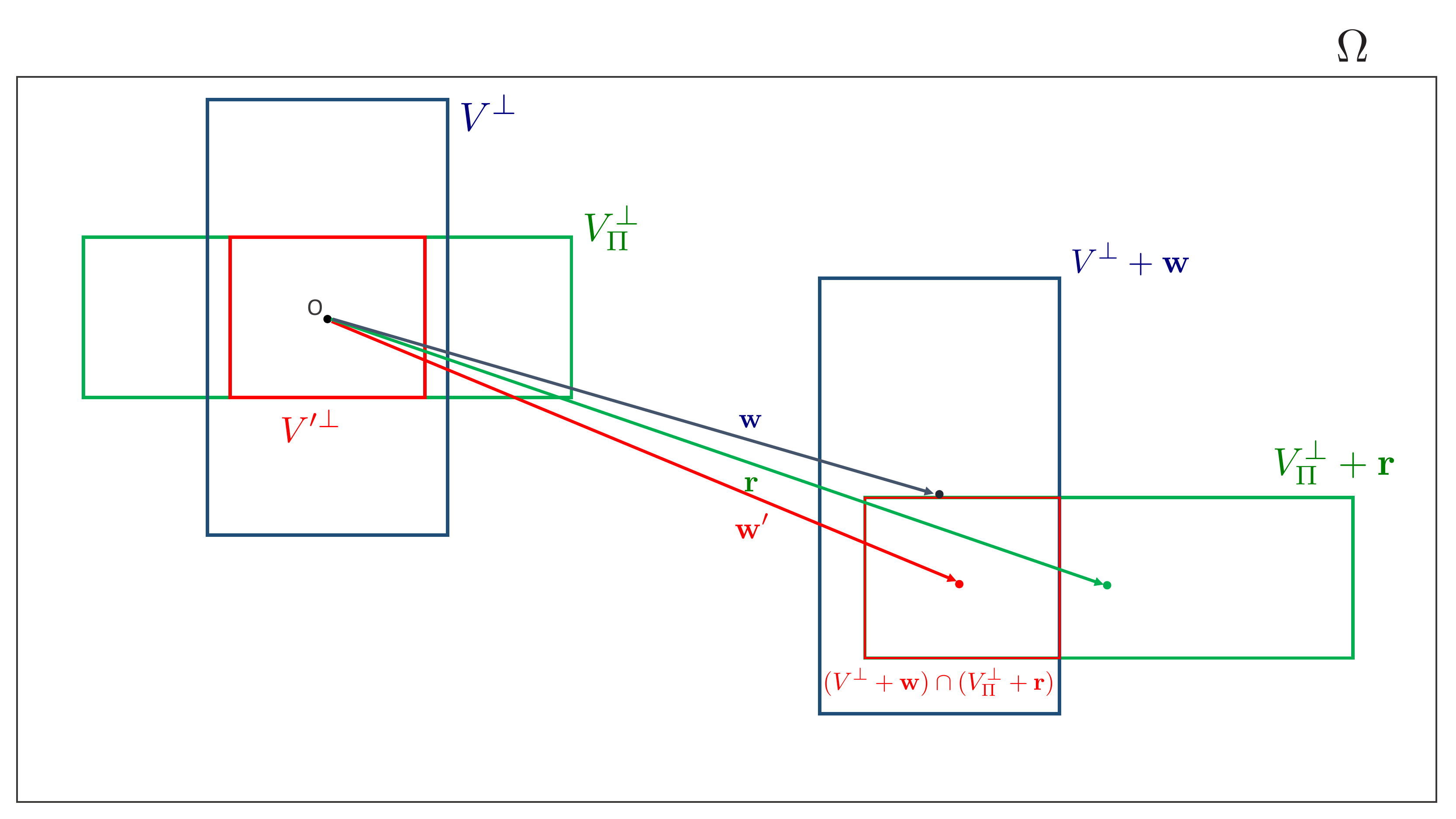}}

\caption{\footnotesize{\textbf{Updating rules via Venn diagrams.} The figure above schematically shows the subspaces $V^{\perp},V_{\Pi}^{\perp},V'^{\perp}$ and the shifted ones (after applying the corresponding representative ontic vectors $\bold{w},\bold{r},\bold{w'}$). In particular this picture explains the expression $V'^{\perp}=(V^{\perp}+\bold{w}-\bold{w'})\cap(V_{\Pi}^{\perp}+\bold{r}-\bold{w'})$ as a result of combining the updating rules for the epistemic subspaces and the representative ontic vectors. It is important to notice that to obtain the correct intersection we have to shift the subspaces $V^{\perp}+\bold{w}$ and $V_{\Pi}+\bold{r}$ back to the same origin (this is the role of $\bold{w'}$). Indeed note that $V^{\perp}\cap V_{\Pi}^{\perp}$ is different from $(V^{\perp}+\bold{w})\cap (V_{\Pi}^{\perp}+\bold{r}).$}}

\label{Venn}
\end{figure*}

\newtheorem{NonCommutingTheorem}[Theorem]{Theorem}
\begin{NonCommutingTheorem}\label{NonCommutingTheorem}
\emph{Non-commuting case}. The epistemic state $(V,\bold{w})$ after a measurement $(V_{\Pi},\bold{r})$ that does not commute with it, \emph{i.e.} some of the generators do not Poisson commute with the state, is described by the epistemic state $(V',\bold{w'})$ such that \begin{equation}\label{SubNonCom}V'^{\perp}=(V_{commute}^{\perp}+\bold{w}-\bold{w'})\cap(V_{\Pi}^{\perp}+\bold{r}-\bold{w'}),\end{equation} where $V_{commute}^{\perp}$ is given by \begin{equation}\label{ultimateBis}V_{commute}^{\perp}=V^{\perp}\oplus V_{other}.\end{equation} The representative ontic vector $\bold{w'}$ is given by \begin{equation}\label{represonticbis}\bold{w'} = \bold{w}+\sum_{i}\bold{\Sigma'}_i^T(\bold{r}-\bold{w})\bold{\gamma}_i,\end{equation} where $\bold{\Sigma'}_i$ are the generators (even the non-commuting ones) of the measurement $\Pi$ and $\gamma_i$ is such that $\bold{\Sigma'}_i^T\bold{\gamma}_i=1.$
\end{NonCommutingTheorem}

\begin{proof}

Let us assume that $\Sigma'_j,$ for $j\in\{0,\dots,m-1\},$ do not commute with the generators of $V.$ In addition to the learning stage of the previous commuting case, we also have a removal stage of the disturbing part of the measurement. We have already seen that we can split the subspace $V$ in $V=V_{commute}\oplus V_{other},$ where $V_{other}$ is generated, from lemma \ref{LemmaNonCommuting}, by all the $\Sigma'_j,$ for $j\in\{0,\dots,m-1\}.$  Therefore we can reduce to the commuting case if we only consider $V_{commute}$ instead of the whole $V.$ The updating rule for the subspace $V$ then becomes \[V\rightarrow V' =V_{commute}\oplus span\{\Sigma'_0,\Sigma'_1,\dots \Sigma'_i,\dots\} =V_{commute}\oplus V_{\Pi}.\]
In terms of the perpendicular subspaces note that we can both write \[V'^{\perp}=(V^{\perp}\oplus V_{other})\cap V_{\Pi}^{\perp},\] and \[V'^{\perp}=V_{commute}^{\perp}\cap V_{\Pi}^{\perp},\] from the usual property that the perpendicular of a direct sum is the intersection of the perpendicular subspaces. The updating rule for the representative ontic vector is the same as in the previous case (equation \eqref{represontic}). The correct updating rule for the subspace $V'^{\perp}$ is found by combining the updating rules for $V$ and $\bold{w}$ as in the previous case \eqref{ultimate}, where $V^{\perp}$ is replaced by $V^{\perp}_{commute}.$ At the end we obtain for the epistemic state $(V',\bold{w'})$ that $V'^{\perp}+\bold{w'}=(V_{commute}^{\perp}+\bold{w})\cap(V_{\Pi}^{\perp}+\bold{r}).$

\end{proof}

Figure \ref{NonCommuting_example} shows a basic example of theorem \ref{CommutingTheorem}.

\begin{figure*}[h!]
\centering

{\includegraphics[width=.8\textwidth,height=.35\textheight]{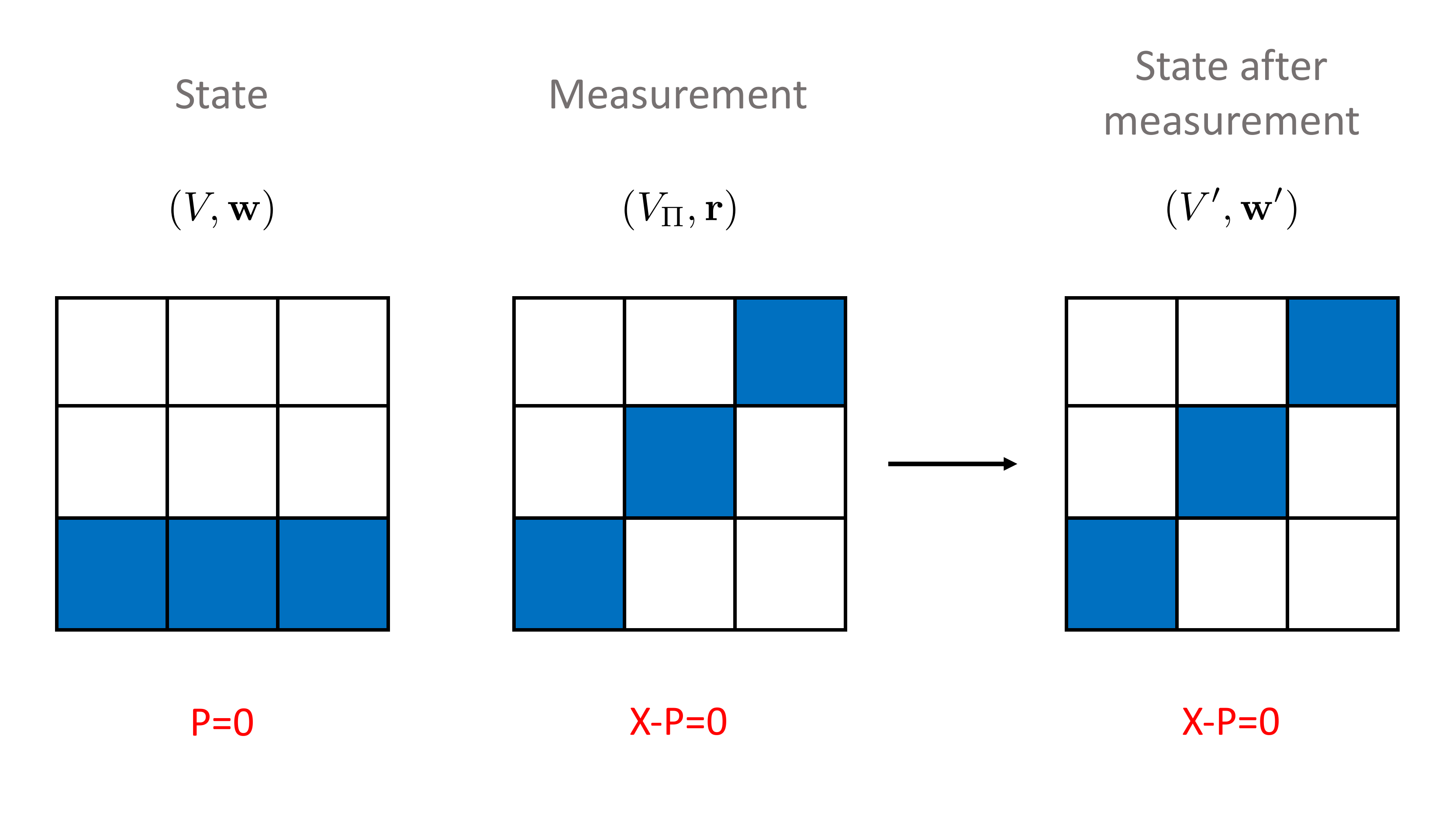}}

\caption{\footnotesize{\textbf{Updating rules in the prime non-commuting case.} The figure above shows a simple one-trit example of theorem \ref{NonCommutingTheorem} regarding the updating rule to predict the state after a sharp measurement that does not commute with the original state. The state after measurement is given by $V'^{\perp}+\bold{w'}=(V_{commute}^{\perp}+\bold{w})\cap(V_{\Pi}^{\perp}+\bold{r}).$ In the above case the shift vectors are all $(0,0),$ the perpendicular subspaces are $V_{commute}^{\perp}=\Omega,$ $V^{\perp}_{\Pi}=span\{(1,1)\},$ and $V'^{\perp}=V^{\perp}_{\Pi}.$}} %Note that with "measurement" we are here representing one element of the measurement. The other elements can be obtained by simply shifting $\bold{r}$ as seen in figure \ref{meas}. The final state is associated to each element of the measurement, each one with a corresponding probability of happening.}}

\label{NonCommuting_example}
\end{figure*}

%%%%%%%%%%%%%%%%%%%%%%%%%%%%%%%%%

% Note that we have derived the updating rules directly from the basic axioms of the theory, \emph{i.e.} the classical complementarity principle.

\section{Updating rules - \emph{non prime} dimensional case}

It is quite common in studies of discrete theories, like Spekkens' model and SQM, to only consider the prime dimensional case because of the particular features of the set of integers modulo $d$, $\mathbb{Z}_d,$ when $d$ is non-prime, like the impossibility of uniquely define inverses of numbers. For example in our present case, figure \ref{figure_example_nonprime} shows the peculiar properties of the observable $3X$ in $d=6,$ which has not full spectrum of outcomes. The general formulation of Spekkens' model of section \ref{SecSpek} does not change; not even the rules for calculating the probabilities of outcome and the updating of the state after a reversible evolutions (which are present in \cite{Spek2}). The new formulation we provide affects the observables and the related measurements updating rules. More precisely our issue, as already noticed, regards the updating-rule formula \eqref{represontic} and \eqref{ultimateBis} for the shift vector $\bold{w'}$ and the subspace $V^{\perp}_{commute},$ which do not always hold when the dimension $d$ is \emph{non-prime}. In fact the vector $\bold{\gamma}_i$ such that $\bold{\Sigma'}_i^T\bold{\gamma}_i=1$ does not always exist in that case. On the other hand, in prime dimensions, it always uniquely exists because $\bold{\gamma}_i=k_i^{-1}\bold{\Sigma'}_i$ and the inverse of the integer $k_i=\bold{\Sigma'}_i^{T}\bold{\Sigma'}_i$ always uniquely exists. 
Unlike the original formulation due to Spekkens, we will now characterise Spekkens' model in non-prime dimensions. In particular we characterise which are the observables that are problematic in the above sense - the \emph{coarse-graining} observables, like $3X$ in $d=6$ - and we then find the updating rules for a state subjected to the measurement of such observables by rewriting them in terms of non-problematic observables - the \emph{fine-graining} observables.

In the next subsection we assume single-system observables (\emph{i.e.} of the kind $\Sigma'=aX+bP,$ $a,b\in\mathbb{Z}_d$) in order to soften the notation and facilitate the comprehension. This will bring more easily to the updating rules even in the most general case of many systems (subsection \ref{bumbum}).
In this case we recall, without making any reference to the quantity $k^{-1},$ but just in terms of the vector $\gamma,$ the updating rule for the shift vector $\bold{w'},$ %We recall the reasoning of section \ref{uprules}. The outcome associated with the observable $\Sigma'$ is $\sigma'.$ \footnote{We assume single system observables (\emph{i.e.} of the kind $aX+bP,$ $a,b\in\mathbb{Z}_d$) for the derivation of the non-prime case in order to soften the notation. The extension to the most general case of many systems will be straightforward and we will point that out in the following.} We assume $\bold{w}$ is not compatible with this outcome, \emph{i.e.} $\bold{\Sigma'}^T\bold{w}=\sigma'+x,$ for some shift $x\in \mathbb{Z}_d,$ and we want to find $\bold{w}'$ such that \begin{equation}\bold{\Sigma'}^T\bold{w'}=\sigma'.\end{equation}
\begin{equation}\label{ShiftGamma}\bold{w}'=\bold{w}-x\bold{\gamma},\end{equation} where, as usual, $x=-\bold{\Sigma'}^T(\bold{r}-\bold{w}),$ and the expression for $V^{\perp}_{commute},$
\begin{equation}\label{CommuteGamma}V^{\perp}_{commute}=\bigcup_c(V^{\perp}+c\gamma).\end{equation}
%Notice that when $k^{-1}$ exists, the vector $\bold{\gamma}$ is unique and it is simply given by $k^{-1}\bold{\Sigma'}.$ However when $k^{-1}$ does not exist, the vector $\bold{\gamma}$ can exist and it is not unique (it shows a certain multiplicity). % vedi after after Dan 33 se vuoi sapere di più!!
%Note that if we were considering more than one generator of the measurement, we would have a summation over all those generators in the second term (see \eqref{summation}). This is simply given by thinking of the whole measurement $\Pi$ as a sequence of measurements given by each generator $\Sigma'_j$ and apply every time the rule \eqref{represontic}.

\subsection{Coarse-graining and fine-graining observables}
% Definizioni. Theorem che fine graining= full spectrum (after dan 34). Theorem che fine graining = esiste gamma (after after dan 33). Come si traduce tutto ciò nelle nostre teorie: \itemise Come scrivere gli uni rispetto agli altri in Spekkens e quindi per Wf.
We define a fine-graining observable as an observable that has \emph{full spectrum}, \emph{i.e.} it can assume all the values in $\mathbb{Z}_d.$ On the contrary a coarse-graining observable has not full spectrum.
\newtheorem{FineTheorem}[LemmaNonCommuting]{Lemma}
\begin{FineTheorem}\label{FineTheorem}
An observable $O_{fg}$ has full spectrum, \emph{i.e.} it is a fine-graining observable, if and only if it has the following form, \begin{equation}\label{fine}O_{fg}=a'X+b'P,\end{equation} where $a',b'\in\mathbb{Z}_d$ are such that they \emph{do not share} any integer factor or power factor of $d.$ %(therefore they are coprime). 
\end{FineTheorem}

On the contrary a coarse-graining observable is written as \begin{equation}\label{coarse}O_{cg}=aX+bP=D(a'X+b'P),\end{equation} where $a',b' \in\mathbb{Z}_d$ are again such that they \emph{do not share} any integer factor or power factor of $d$ and $D$ is a factor shared by $a,b\in\mathbb{Z}_d.$ More precisely the factor $D$ is called \emph{degeneracy}  and it is defined as \begin{equation}\label{degeneracy}D=D_1^{n_1}\cdot D_2^{n_2} \cdot \dots ,\end{equation} where $D_1,D_2,\dots$ are different integer factors of $d$ shared by $a$ and $b,$ and $n_1,n_2,\dots$ are the maximum powers of these factor such that they can still be grouped out from $a$ and $b$. We take the maximum powers because we want the remaining part, $a'X+b'P,$ to not share any common integer factor or power factor of $d$ between $a'$ and $b'.$ In this way we can associate a fine-graining observable to a coarse graining one by simply dropping the degeneracy $D$ from the latter.

\begin{proof}

Let us first prove that an observable of the kind \eqref{fine}, $O_{fg}=a'X+b'P,$ is a full spectrum one. 
This can be proven by using Bezout's identity \cite{Bezout}: let $a'$ and $b'$ be nonzero integers and let $D$ be their greatest common divisor. Then there exist integers $X$ and $P$ such that $aX+bP=D$. In our case the greatest common divisor $D$ is equal to one, since $a',b'$ are coprime. \footnote{ It could be that $a',b'$ share a factor which is not a factor of $d.$ In this case the argument follows identically as if they were coprime.} Therefore we have proven that there exist values of the canonical variables $X,P\in\mathbb{Z}_d$ such that $O_{fg}=a'X+b'P=1.$ In order to reach all the other values of the spectrum we simply need to multiply both $X$ and $P$ in the previuos equation by $j\in\mathbb{Z}_d.$

We now prove the converse, \emph{i.e.} that a full spectrum observable implies it to be written as \eqref{fine}. We prove this by seeing that an observable written as \eqref{coarse} has not full spectrum, \emph{i.e.} we negate both terms of the reverse original implication. Proving the latter is straightforward, since the multiplication modulo $d$ between an arbitrary quantity and a factor $D,$ which is given by powers of integer factors of $d,$ gives as a result a multiple of $D.$ Since the multiples of $D$ do not cover the whole $\mathbb{Z}_d,$ then any observable of the form \eqref{coarse} has not full spectrum. \footnote{Multiples of $D$ do not cover the whole spectrum of $\mathbb{Z}_d$ because D has not an inverse $D^{-1}$  (it is not coprime with $d$) and so we cannot obtain the whole values $\sigma$ of $\mathbb{Z}_d$ by simply finding $X,P$ such that $a'X+b'P=D^{-1}\sigma.$} 
 Since an observable of the form \eqref{fine} is an observable that cannot be written as \eqref{coarse} by definition, we obtain that a full spectrum observable implies the observable to be written as \eqref{fine}.

\end{proof}

Given lemma \ref{FineTheorem} we have got the expressions \eqref{coarse} and \eqref{fine} for coarse-graining and fine-graining observables. We want now to prove the following lemma to ensure that fine-graining observables are characterised by precisely defined updating rules.

\newtheorem{GammaTheorem}[LemmaNonCommuting]{Lemma}
\begin{GammaTheorem}\label{GammaTheorem}
The vector $\bold{\gamma}$ in the updating rule \eqref{ShiftGamma} for the shift vector $\bold{w'}$ and in the equation \eqref{CommuteGamma} for the subspace $V^{\perp}_{commute}$ exists if and only if the observable is a fine-graining one.
\end{GammaTheorem}
  
\begin{proof}
Let us prove that if we have a fine graining observable the vector $\bold{\gamma}$ exists. In our case $\Sigma'=(a',b')$ and, by definition of $a',b'$ (as usual defined for fine-graining observables) and full spectrum, we can always find a vector $\bold{\gamma}=(\gamma_a,\gamma_b)$ such that  $\bold{\Sigma'}^{T}\bold{\gamma}=a'\gamma_a + b'\gamma_b$ equals $1.$

Let us prove the converse. We now have the vector $\bold{\gamma}$ such that $\bold{\Sigma'}^{T}\bold{\gamma}=a\gamma_a + b\gamma_b=1,$ where the coefficients $a,b\in\mathbb{Z}_d$ define our observable $aX+bP=\sigma.$ We want to prove that $\sigma$ can achieve all the values of $\mathbb{Z}_d.$ Since  $\bold{\Sigma'}^{T}\bold{\gamma}=1$ we can set the values of $(X,P)$ as equal to $(\gamma_a,\gamma_b)$ in order to reach the value $\sigma=1.$ We can now achieve all the other values of the spectrum by simply redefining $\gamma$ as $\tilde{\bold{\gamma}}=c \bold{\gamma},$ where $c$ assumes all the values in $\mathbb{Z}_d.$ 

\end{proof}  

The above lemma \ref{GammaTheorem} should convince us that in order to find the updating rules in the presence of a coarse graining observable, it is appropriate to decompose it in terms of fine-graining observables. 
Let us assume that our coarse-graining observable is $O_{cg}=aX+bP=D(a'X+b'P)=\sigma$, and the associated isotropic subspace and representative ontic vector are $(V_{cg},\bold{r}_{cg}).$ To this observable we can associate $\bar{D}$ different fine-graining observables $O_{fg}=a'X+b'P=\sigma_j,$ where $j\in{0,\dots,\bar{D}-1}.$ The quantity $\bar{D}$ is the degeneracy $D$ without the powers $n_1,n_2,\dots,$ \emph{i.e.} $\bar{D}=D_1\cdot D_2 \cdot \dots.$ Indeed the powers $n_1,n_2,\dots$ simply represent multiplicities associated to each corresponding fine-graining observable. The associated isotropic subspaces and representative ontic vectors are $(V_{fg},\bold{r}^{(j)}_{fg}),$ where $V_{fg}=span\{(a',b')\}$ (see figure \ref{figure_example_nonprime}).

\begin{figure*}[h!]
\centering

{\includegraphics[width=.9\textwidth,height=.38\textheight]{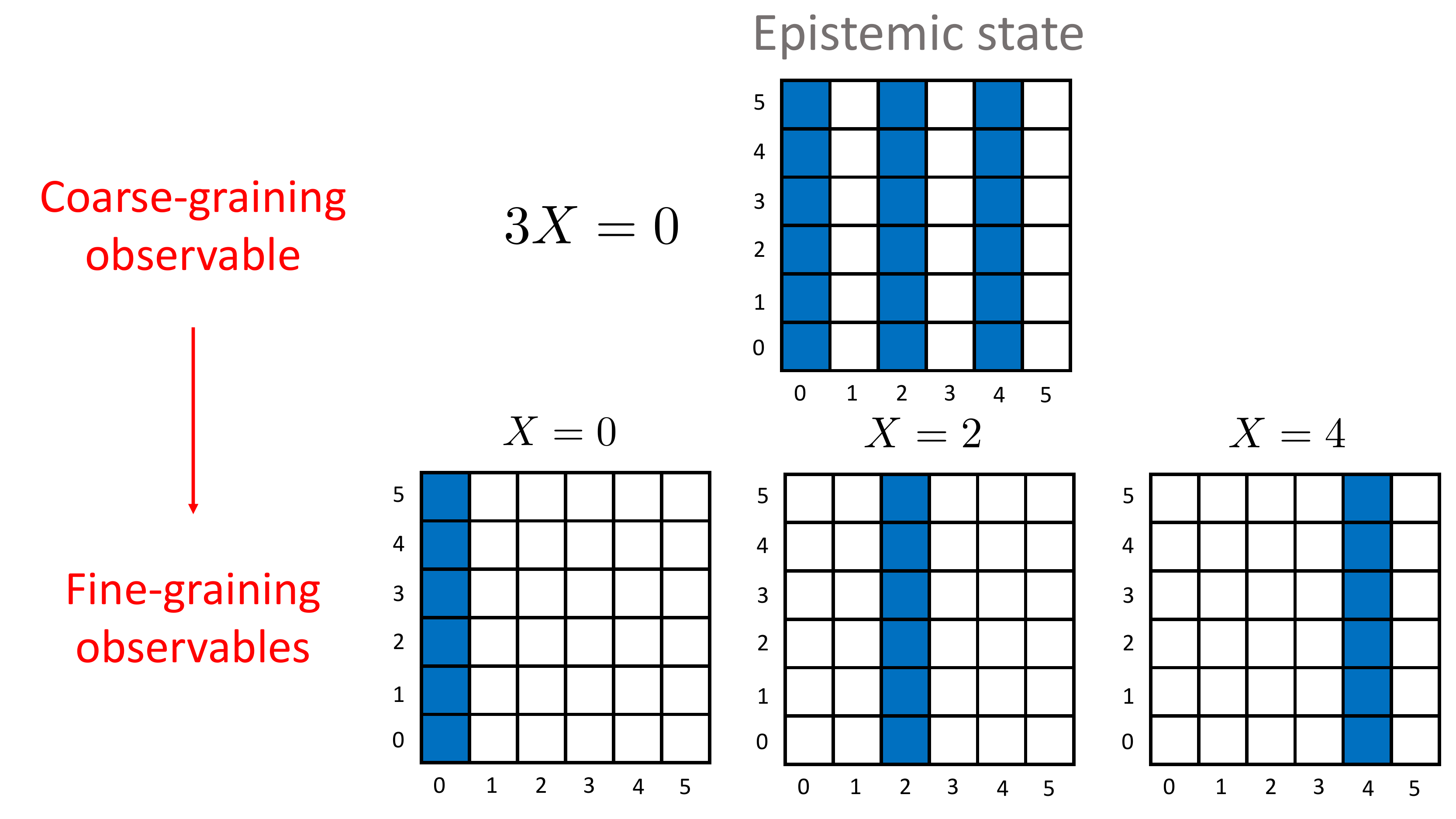}}

\caption{\footnotesize{\textbf{Simple example of a coarse-graining observable and its decomposition in fine-graining observables in $d=6$.} The coarse-graining observable $O_{cg}=3X=0$ in $d=6$ shows degeneracy $D=3.$ The three fine-graining observables associated with $O_{cg}$ are $O^{(0)}_{fg}=X=0,$ $O^{(1)}_{fg}=X=2$ and $O^{(2)}_{fg}=X=4.$ %The anti-degeneracy is $C=2,$ indeed $C\cdot D=6=0.$ 
The perpendicular subspaces of known variables are $V^{\perp}_{cg} = span\{(0,1),(2,0)\},$ $V^{\perp}_{fg} = span\{(0,1)\}$ and $V_D=span\{(2,0)\}.$ A choice for the representative ontic vectors is $\bold{r}_{cg}=(0,0),$ $\bold{r}^{(0)}_{fg}=(0,0),$ $\bold{r}^{(1)}_{fg}=(2,0)$ and $\bold{r}^{(2)}_{fg}=(4,0).$
Notice that not all the values are possible for the coarse-graining observable $3X$ to be a valid observable. Only $3X=0$ and $3X=3$ are valid (indeed what would it be the epistemic state representation for \emph{e.g} $3X=2?$), as witnessed by the expression \eqref{errej} for the associated fine-graining observables, that is valid only when the ratio $\frac{\sigma_{cg}}{D}$ exists.
}}

\label{figure_example_nonprime}
\end{figure*}

By definition the perpendicular isotropic subspaces are \begin{equation}\label{SubCoars} V^{\perp}_{cg} =\{\bold{v}=(v_a,v_b)\in\Omega | v_a a + v_b b  =D(v_a a' + v_b b')=0 \; \text{mod}(d)\}\end{equation} \begin{equation}\label{SubFine}V^{\perp}_{fg}=\{\bold{v'} =(v'_a,v'_b)\in\Omega | v'_a a' + v'_b b'=0 \; \text{mod}(d)\}.\end{equation} %=\bigcup^{D-1}_{j=0}V^{(j)\perp}_{fg}
It is clear that $V^{\perp}_{cg}\supset V^{\perp}_{fg}$ and we can therefore construct $V^{\perp}_{cg}$  as \begin{equation}\label{FineSpek}V^{\perp}_{cg} = \bigcup^{\bar{D}-1}_{j=0}(V^{\perp}_{fg}+\bold{v}_j) =V^{\perp}_{fg}\oplus V_D,\end{equation} where the subspace $V_D$ provides all the vectors that we need to combine with the vectors of  $V^{\perp}_{fg}$ to reach the whole $V^{\perp}_{cg}.$ We call the subspace $V_D$ the \emph{degeneracy subspace} because it encodes the degeneracy of $V_{cg}$ with respect to $V_{fg}$ . It has dimension $1$ and size $\bar{D}.$ %\footnote{The reason why we do not consider the powers $n_1,n_2,\dots$ is because in order to get all the vectors of $V_D$ this would give us redundancy, \emph{i.e.} if a vector of $V_D$ is given by \emph{e.g.} $C_1\bold{v},$ then the vector $C_1^{n_1}\bold{v}$ will be already given by another factor of $\bar{D}$ multiplying $\bold{v},$ for example $C_2\bold{v}.$ This happens because $\bold{v}$ is defined as $C\bold{\Sigma}_{fg},$ where $C$ is the antidegeneracy. If we consider all the $i\bold{v},$ where $i\in\{0,1,\dots D\},$ then every $i$ can be written as $i=j\text{mod}(\bar{D}),$ where $j\in \{0,1,\dots,\bar{D}\},$ thus showing redundancy.} 
%This is always the case because in general any $C_j^{n_j}$ is not a factor of $d,$ but it is a factor of $m\cdot d,$ for some integer $m.$ 
 This is consistent with the fact that the dimensions of $V^{\perp}_{cg}$ and $V^{\perp}_{fg}$ are respectively $2$ and $1.$ The sizes are respectively $\bar{D}\cdot d$ and $d.$ %This is consistent with the fact that  the number of ways to obtain $C\neq 0$ given $\bold{v}=(v_a,v_b)$ are $d$ (size of $V^{\perp}_{fg}$ ). 
The size of $V^{\perp}_{fg}$ is $d$ because it is always a maximally isotropic subspace and its dimension is $1$ because from one generator we get all the other vectors of the subspace by multiplication with $j\in\mathbb{Z}_d$. The dimension $V^{\perp}_{cg}$ is $2$ because it cannot be $1$ (it would be the same subspace as $V^{\perp}_{fg}$) and it cannot be greater than $2$ since also the whole phase space $\Omega=\mathbb{Z}^{2}_d$ has dimension $2.$ In order to know the size of $V^{\perp}_{cg}$ we need to count all the $j\bold{v},$ where $j\in \{0,1,\dots,\bar{D}-1\},$ that means $\bar{D}\cdot d.$
Therefore it can be written as $V_D=span\{\bold{v}\},$ and all its $\bar{D}$ vectors are of the kind $\bold{v}_j=j \bold{v}.$ The above reasoning easily extends to the case of $n$ systems, where the dimensions are $\text{dim}(V^{\perp}_{cg})=2n, \; \text{dim}(V^{\perp}_{fg})=n, \; \text{dim}(V_{D})=n,$ and the sizes are $|V^{\perp}_{cg}|=\bar{D}^{n}d^{n}, \; |V^{\perp}_{fg}|= d^{n}, \; |V_{D}|=\bar{D}^{n}.$ 
We can now prove that $V_D$ is a vector space. 
\begin{proof}
The definition of $V_D$ is \begin{equation}\label{VD} V_{D} =\{\bold{v}\in\Omega | \alpha \bold{w} + \beta \bold{v}=\bold{t}, where 
\; \bold{w}\in V^{\perp}_{fg}, \; \alpha,\beta\in\mathbb{Z}_d,\; \bold{t}\in V^{\perp}_{cg}\}.\end{equation}
To see that it is a vector space we just need to see that $(0,0)$ belongs to $V_D$ and that $V_D$ is closed under addition and multiplication, \emph{i.e.} under linear combinations. The null vector belongs to $V_D$ because in the definition \eqref{VD} we would remain with $\alpha\bold{w}=\bold{t},$ where $\bold{w}\in V^{\perp}_{fg}$ and $V^{\perp}_{fg} \subset  V^{\perp}_{cg}.$  Let us imagine that we have two vectors  $\bold{v},\bold{z}\in V_{D}.$ Is the vector $\gamma \bold{v}+ \delta \bold{z},$ where $\gamma,\delta\in \mathbb{Z}_d,$ still belonging to $V_D?$ It is easy to see that if we apply the definition \eqref{VD} we would get \[\alpha\bold{w}+\beta(\gamma \bold{v}+ \delta \bold{z}),\] which can be rewritten as 
\[( \alpha \bold{w}+\beta\gamma\bold{v})+ (0\cdot \bold{w}+ \beta\delta\bold{z}),\] where each of the two terms in parenthesis belong to $V^{\perp}_{cg},$ and therefore the whole expression belongs to it too.
\end{proof}

We now define the shift vectors $\bold{r}^{(j)}_{fg}$ in terms of $\bold{r}_{cg}$ and see that we can encode the degeneracy expressed by $V_D$ in there. The idea is schematically depicted in figure \ref{VennNonPrime}.

\begin{figure*}[h!]
\centering

{\includegraphics[width=.99\textwidth,height=.45\textheight]{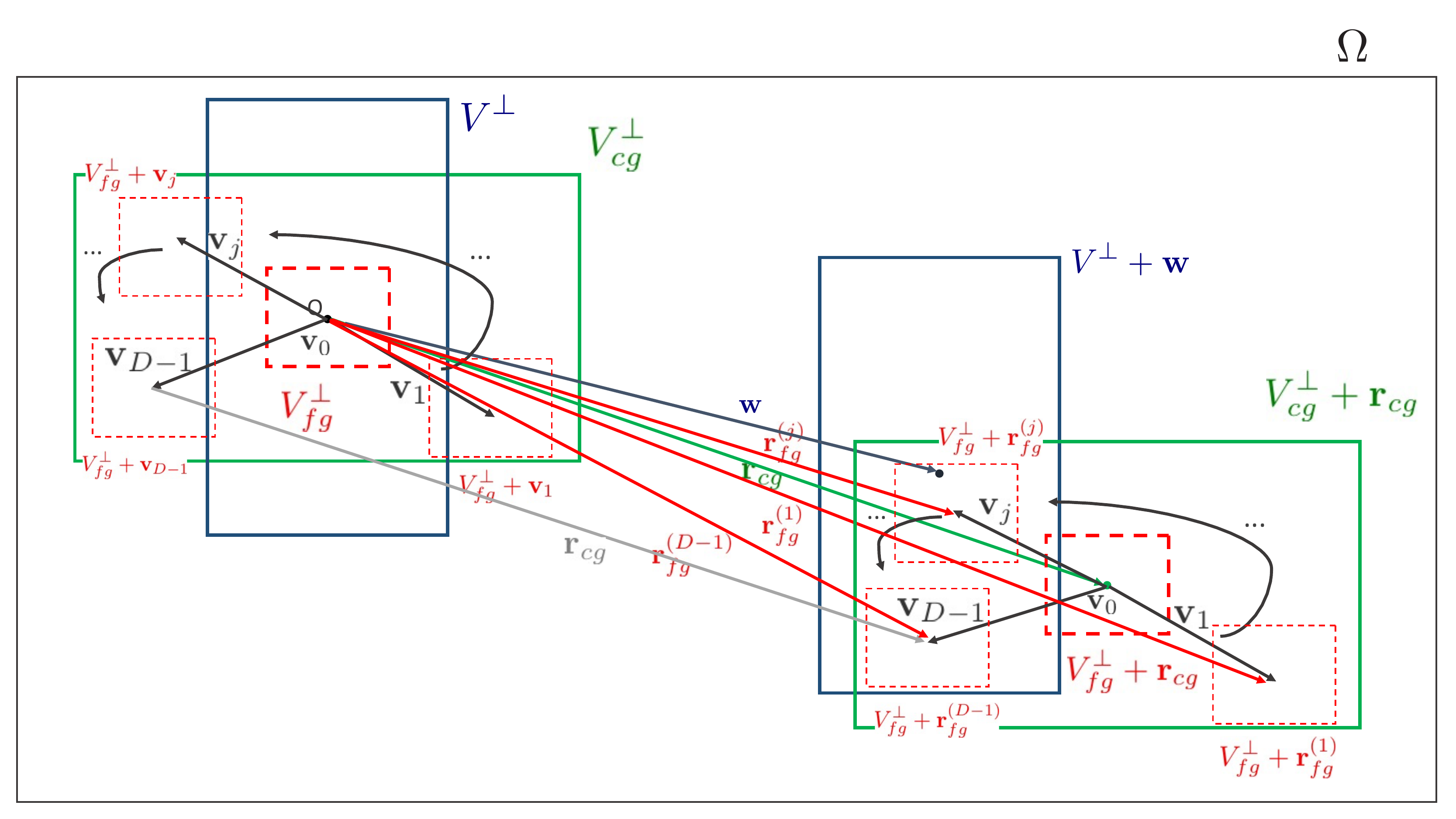}}

\caption{\footnotesize{\textbf{Schematic representation of Coarse-graining decompositions into fine-graining observables.} The figure above schematically represents the relation between the subspaces $V^{\perp},V^{\perp}_{cg}, V^{\perp}_{fg}$ and their corresponding shift vectors $\bold{w},\bold{r}_{cg},\bold{r}^{(j)}_{fg}.$ The green rectangles represent the subspaces $V^{\perp}_{cg}$ and its translated $V^{\perp}_{cg}+\bold{r}_{cg}.$ The latter can be seen to be equivalent either to the dashed red rectangle representing $V^{\perp}_{fg}$ shifted by the degeneracy vectors of $V_{D}$ (light black arrows), this corresponding to $V^{\perp}_{cg},$ and then shifted by $\bold{r}_{cg}$ (green arrow or light grey arrow), or to the dashed rectangle representing $V^{\perp}_{fg}$ shifted by each $\bold{r}^{(j)}_{fg}$ (red arrows). Both are in accordance with the expressions of $V^{\perp}_{cg}+\bold{r}_{cg},$ $V^{\perp}_{cg}+\bold{r}_{cg}=V^{\perp}_{fg}\oplus V_{D}+\bold{r}_{cg}=V^{\perp}_{fg}+\sum_{j=0}^{\bar{D}-1}(\bold{r}_{cg}+j\bold{v}),$ where $\bold{v}$ is the generator of $V_{D}.$ Note that, as a consequence of the degeneracy characterising the coarse-graining observable, we could keep adding $\sum_{j=0}^{\bar{D}-1} j\bold{v}=D\bold{v}$ without changing the validity of the expression of $V^{\perp}_{cg}+\bold{r}_{cg}.$ We can see it just by noticing that $V_D +\sum_{j=0}^{\bar{D}-1} j\bold{v}=V_{D}$ or, more simply, that $D\bold{v}=D\cdot C\bold{\Sigma'_{fg}}=0,$ since $D\cdot C =0 \text{mod}(d).$ 
}}

\label{VennNonPrime}
\end{figure*}

Given the shift vector associated to the coarse-graining observable $\bold{r}_{cg},$ the shift vectors $\bold{r}^{(j)}_{fg}$ associated to the corresponding fine-graining observables are of the kind \begin{equation}\label{ShiftFine}\bold{r}^{(j)}_{fg}=\bold{r}_{cg}+\bold{v}_j, \end{equation} where $\bold{v}_j\in V_d$ and are therefore of the kind $j\bold{v},$ where $j\in\{0,\dots,\bar{D}-1\}.$ This implies that if we assume the outcome associated to the coarse-graining observable to be $\sigma_{cg},$ \emph{i.e.} $\bold{\Sigma}^{T}_{cg}\bold{r}_{cg}=\sigma_{cg},$ where $\Sigma_{cg}=(a,b),$ then the outcomes associated to the fine graining-observables are \begin{equation}\label{errej}\bold{\Sigma}^{T}_{fg}\bold{r}^{(j)}_{fg}=\bold{\Sigma}^{T}_{fg}(\bold{r}_{cg}+j\bold{v})=\frac{\sigma_{cg}}{D}+jC,\end{equation} where $C$ is the \emph{anti-degeneracy} and it is defined as a non-zero number belonging to $\mathbb{Z}_d$ such that $D\cdot C=0  \; \text{mod}(d).$ The idea is that the vector $\bold{v}\in V_D$ is such that $\bold{\Sigma}^{T}_{fg}\bold{v}=C\neq 0,$ so it does not belong to $V^{\perp}_{fg},$ but it does belong to $V^{\perp}_{cg},$ since $D\cdot C=0 \; \text{mod}(d).$ An easy way to find one of the possible $\bold{v}$ is to calculate it as $C\bold{\Sigma}_{fg},$ where $\Sigma_{fg}$ is the generator of $V_{fg}.$  In this way we know that $D\bold{v}=0,$ but $\bold{v}$ does not belong to $V^{\perp}_{fg},$ \emph{i.e.} $v_a a' + v_b b'\neq 0$ because $\Sigma_{fg}$ is not in $V_{fg}^{\perp}.$
It is important to notice that equation \eqref{errej} implies that not all the outcomes are allowed for the fine-graining observables associated to the coarse-graining one; they are allowed only when the ratio $\frac{\sigma_{cg}}{D}$ exists. Figure \ref{figure_example_nonprime} also explains this fact.

\subsection{Measurement updating rules} \label{bumbum}

Let us assume to have $n$ systems and to measure the coarse-graining observable $O_{cg}=a_1X_1+b_1P_1+\dots +a_nX_n+b_nP_n=D(a'_1X_1+b'_1P_1+\dots +a'_nX_n+b'_nP_n)=\sigma_{cg},$ with corresponding isotropic subspace of known variables $V_{cg}$ and shift vector $\bold{r}_{cg},$ on the state $\rho=\alpha_1X_1+\beta_1P_1+\dots +\alpha_nX_n+\beta_nP_n=\sigma,$ with corresponding isotropic subspace of known variables $V=span\{\Sigma_1,\dots,\Sigma_n\}$ and shift vector $\bold{w}$.
The idea in order to find the updating rules for the state after measurement, the subspace of known variable $V'$ and the representative ontic vector $\bold{w'}$ is to compute the updating rule of the initial state $\rho$ with the fine-graining observables that are associated to the coarse graining observable $O_{cg},$ \emph{i.e.} $O^{(j)}_{fg}=a'_1X_1+b'_1P_1+\dots +a'_nX_n+b'_nP_n=\sigma^{(j)}_{fg}$ (indeed we know that the updating rules are valid for them from theorem \ref{GammaTheorem}), and then combine them together. More precisely, the following theorem holds.

\newtheorem{FinalbisTheorem}[Theorem]{Theorem}
\begin{FinalbisTheorem}\label{FinalbisTheorem}
%In the case of $n$ systems observables, 
The epistemic state $(V,\bold{w})$ after a coarse-graining measurement $(V_{cg},\bold{r}_{cg})$ is described by the epistemic state $(V',\bold{w'})$ such that \begin{equation}\label{utlimateBoom} V'^{\perp}=\bigcup^{\bar{D}-1}_{j=0} [(V_{commute}^{\perp}+\bold{w}-\bold{w'})\cap(V^{\perp}_{fg}+\bold{r}^{(j)}_{fg}-\bold{w'})],\end{equation} where the shift vector $\bold{w'}$ is the shift vector deriving from the updating rule of the state after the measurement of the fine-graining observable $O^{(j)}_{fg},$ \begin{equation}\bold{w'}=\bold{w'}_j=\bold{w}+\sum_{i=0}^{n}\bold{\Sigma'}^T_{i}(\bold{r}^{(j)}_{fg}-\bold{w})\gamma_i,\end{equation} where the vectors $\gamma_i$ are defined such that $\bold{\Sigma'}^T_i\gamma_i=1,$ and $\Sigma'_i$ are the $n$ generators of the subspace $V_{fg}$ associated to the fine-graining  observable $O^{(j)}_{fg}.$ The subspace $V_{commute}^{\perp}$ is given by the original $V$ after having removed the non-commuting part, \emph{i.e.} equation \eqref{ultimateBis}, \begin{equation}V_{commute}^{\perp}=\bigcup^{d}_{c=1}(V^{\perp}+\sum_{l=N+1}^{n}c\gamma_l)=V^{\perp}\bigoplus_{l=N+1}^{n} V_l = V^{\perp}\oplus V_{other},\end{equation} where $\gamma_l$ is such that $\bold{\Sigma}_l^{T}\gamma_l=1$ and $V_l$ are the subspaces spanned by the $(n-N)$ non-commuting generators $\Sigma_l$. Obviously if the state and measurement commute, then $V_{commute}^{\perp}=V^{\perp}.$ 
\end{FinalbisTheorem}

The above theorem tells us that the way we combine the updating subspaces of the state with each individual fine-graining observables is through their union. This result is clear in terms of schematic diagrams (figure \ref{VennNonPrime}). The updated shift vector is just one of the updated shift vectors of the state with the fine-graining observables, because the information needed to update the shift vector of the state is encoded in just one of the fine-graining shift vectors. The degeneracy includes a meaningless multiplicity in the coarse-graining shift vector, and therefore every fine-graining observable can do the job of correctly updating the shift vector of the state. Actually every combination of the shift vectors $\bold{w'}_{j}$ can do the job, apart from the ones that sum to $0 \; \text{mod}(d),$ like $\sum_{j=0}^{\bar{D}-1}\bold{w'}_j.$ 
Note also that in the definition of $V^{\perp}_{commute}$ the vector $\gamma_l$ is, in general, degenerate. This is not a problem because any degenerate value of $\gamma_l$ brings to the same subspace $V^{\perp}_{commute},$ since by definition its role is to add the vectors $\lambda'=c\gamma_l$ to $V^{\perp}$ such that $\Sigma_n^T\lambda'\neq0.$ %It's clear by definition of $\gamma$ that any degenerate value of $\gamma$ does this job.   

%This result is not correctly represented by the Venn diagrams in figure \ref{VennNonPrime}, since the addition of other vectors would enlarge the sets in the figure. (CAN WE FIND A BETTER DIAGRAMMATIC REPRESENTATION OF DEGENERACY?)
 %Another feature that is not present in the Venn representation of figure \ref{VennNonPrime} is that if we keep adding $\sum_{j=0}^{D-1} \bold{v}_j$ to $\cup_{j=0}^{D-1}(V^{\perp}_{fg}+\bold{r}^{(j)}_{fg})$ nothing changes in order to have the equivalence with $V^{\perp}_{cg}+\bold{r}_{cg}.$
%Note that equation \eqref{utlimateBoom} is the same as \eqref{ultimate} of the prime-dimensional case, where the only difference is that we take the union of the subspaces of fine-graining observables that decompose the subspace of the coarse-graining one. Moreover

\begin{proof}
%The proof of the updating rule for the subspace of known variables ($V^{\perp} \rightarrow V'^{\perp}$) is obtained just by considering the decomposition of the coarse-graining observable into fine-graining observables. 
We find the expression for the updated subspace $V'^{\perp}$ by simply reusing the already found formulas \eqref{ultimate} and \eqref{ultimateBis} of the prime-dimensional case and substituting $V^{\perp}_{\Pi}$ with $V^{\perp}_{cg}$ and $\bold{r}$ with $\bold{r}_{cg},$ \[V'^{\perp}=(V_{commute}^{\perp}+\bold{w}-\bold{w'})\cap(V^{\perp}_{cg}+\bold{r}_{cg}-\bold{w'}).\] If we now consider the decomposition of $V^{\perp}_{cg}+\bold{r}_{cg}$ as in  \eqref{FineSpek} and \eqref{ShiftFine}, we obtain \[V'^{\perp}=(V_{commute}^{\perp}+\bold{w}-\bold{w'})\cap[\cup^{\bar{D}-1}_{j=0} (V^{\perp}_{fg}+\bold{r}^{(j)}_{fg})-\bold{w'}].\] Since the intersection of a union is the union of the intersections, we have proven the first part of the theorem, \[V'^{\perp}=\bigcup^{\bar{D}-1}_{j=0} [(V_{commute}^{\perp}+\bold{w}-\bold{w'})\cap(V^{\perp}_{fg}+\bold{r}^{(j)}_{fg}-\bold{w'})].\] 
The second part of the proof regards $\bold{w'}$ being equal to any of the $\bold{w'}_j.$  Because of the degeneracy, any $\bold{w'}_j$ is equivalent to the others (with different value of $j$) in order to provide us with $\bold{w'}$, indeed it is possible to find one from another just by adding a vector $\bold{v}\in V_D.$ The latter can be proven as follows. For simplicity let us assume to be in the case $n=1$ and that $\bold{v}$ is the generator of $V_D.$ We know that, by the definition of state after measurement of a fine-graining observable, the updated shift vector $\bold{w'}_j$ is such that $\bold{\Sigma}^{T}_{fg}\bold{w'}_j=\frac{\sigma_{cg}}{D}+jC=\sigma^{(j)}_{fg},$ where $C=\bold{\Sigma}^{T}_{fg}\bold{v}$ is the antidegeneracy (equation \eqref{errej}). It is straightforward to see that if we add $\bold{v}$ to $\bold{w'}_j,$ we get $\bold{w'}_j+\bold{v}=\bold{w'}_{j+1},$ indeed $\bold{\Sigma}^{T}_{fg}(\bold{w'}_j+\bold{v})=\frac{\sigma_{cg}}{D}+(j+1)C=\sigma^{(j+1)}_{fg}.$ %The question now is how to prove that $\bold{w'}=\bold{w'}_{j},$ where, according to \eqref{represontic}, $\bold{w'}_{j}=\bold{w} + \Sigma_{fg}^T (\bold{r}^{(j)}_{fg} ? \bold{w})\bold{\gamma}_{fg},$ where $\bold{\gamma}_{fg}$ is such that $\Sigma_{fg}^T\bold{\gamma}_{fg}=1.$ 
\end{proof}

Figure \ref{NonPrime_example} shows a basic example of theorem \ref{FinalbisTheorem}.

\begin{figure*}[h!]
\centering

{\includegraphics[width=.8\textwidth,height=.35\textheight]{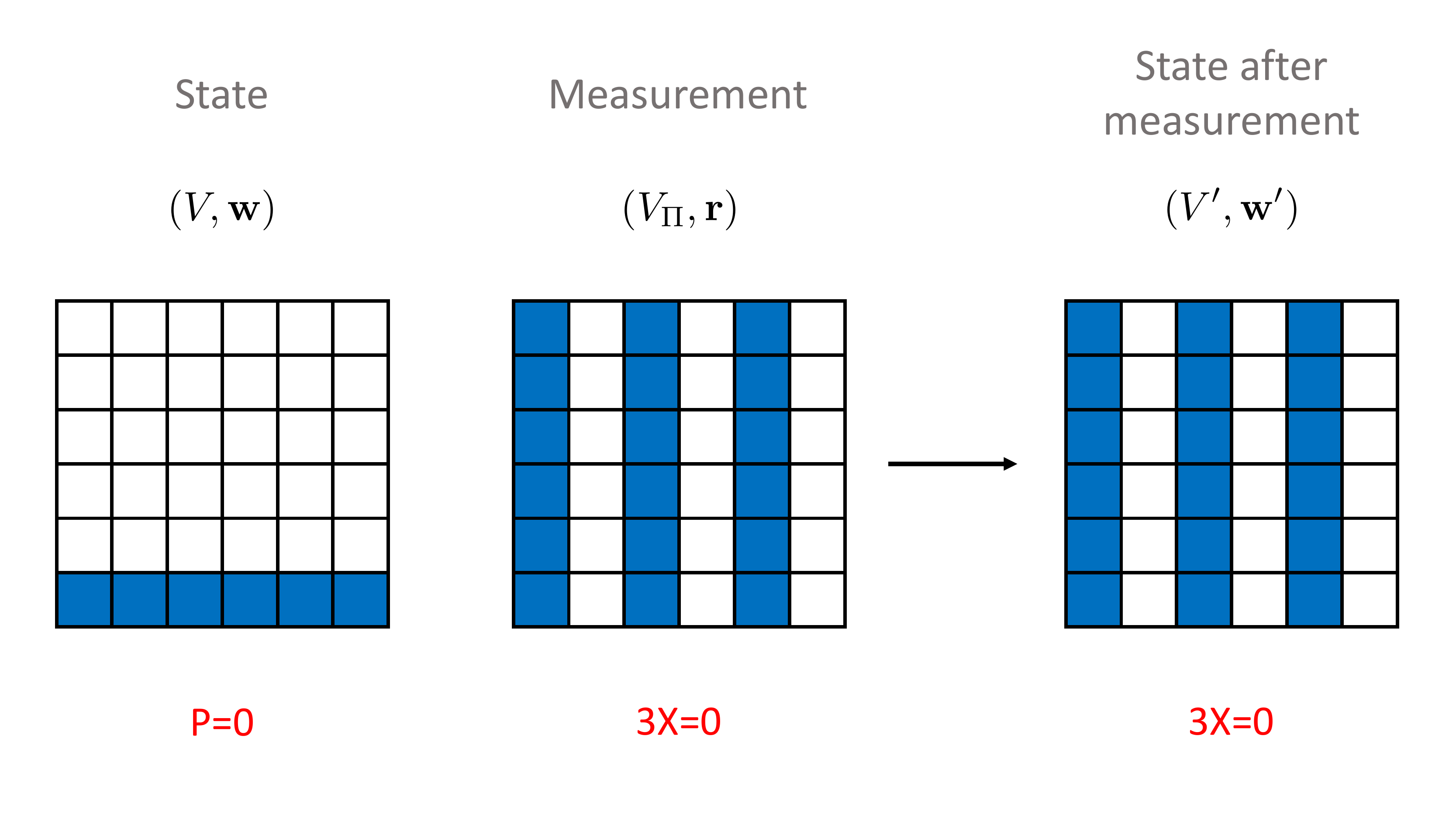}}

\caption{\footnotesize{\textbf{Updating rules in the non-prime non-commuting case.} The figure above shows a simple example (one system in $d=6$) of theorem \ref{FinalbisTheorem} regarding the updating rule to predict the state after a sharp measurement that does not commute with the original state. The state after measurement is given by $V'^{\perp}+\bold{w'}=\bigcup^{\bar{D}-1}_{j=0} [(V_{commute}^{\perp}+\bold{w})\cap(V^{\perp}_{fg}+\bold{r}^{(j)}_{fg})].$ In the above case the shift vectors are $\bold{w}=(0,0),\bold{r}^{(0)}_{fg}=(0,0),\bold{r}^{(1)}_{fg}=(2,0),\bold{r}^{(2)}_{fg}=(4,0),\bold{w'}=(0,0),$ the perpendicular subspaces are $V_{commute}^{\perp}=\Omega,$ $V^{\perp}_{fg}=span\{(0,1)\},$ $V^{\perp}_{\Pi}=span\{(0,1),(2,0)\},$ and $V'^{\perp}=V^{\perp}_{\Pi}.$}} %Note that with "measurement" we are here representing one element of the measurement. The other elements can be obtained by simply shifting $\bold{r}$ as seen in figure \ref{meas}. The final state is associated to each element of the measurement, each one with a corresponding probability of happening.}}

\label{NonPrime_example}
\end{figure*}

%The above theorem can be restate for the general case of many systems, and so arbitrary many generators of the coarse-graining measurement, as follows

%Notice that all the results found in the non-prime dimensional case and only proven for the case of one system, also hold when considering many systems. In particular for the updating rule of the representative ontic vector in the case on non-commuting measurements (theorem \ref{FinalTheorem}), one may think $\bold{w'}$ to modify some of the variables that should not change (among the ones that belong to observables that commute with the state). However, by construction, the shift vector $\bold{w'}$ represents the measurement (indeed $\Sigma_{fg}^{T}\bold{w'}=\sigma^{j}_{fg}$) and therefore commute with state in all the parts that commute with the measurement itself, \emph{i.e.} the parts that do not need to change by updating the shift vector, do not change. The same reasoning held in the prime dimensional case.

%HERE I need to write a part about the multiplicity of the shift vectors and on $\gamma$.

%!!!!!!!!!!!!!!!!!!!!!!!!!!!!!!!!!!!!!!!!!!!!!!!!

% da qua in poi non funziona piu footnotes. Mistero della fede. Il motivo è che si sono finiti i simboli perchè non le numerano. Che ritardati.

\section{Equivalence of Spekkens' theory and SQM in all odd dimensions}

In \cite{Spek2} it has been shown that SQM and Spekkens' toy model are two operationally equivalent theories in odd prime dimensions via Gross' theory of discrete non-negative Wigner functions. We have generalised Spekkens' model to all discrete dimensions. The above equivalence does not hold in even dimensions, but we will now see that it holds in \emph{all} odd dimensions. We will also state the equivalence in terms of the updating rules, where all its elegance arises.  
We recall that SQM and Gross' theory of non-negative Wigner functions are equivalent in all odd dimensions \cite{Gross}. %We first state the updating rules in SQM and then we find them in Gross' theory from the Spekkens' ones. 

\subsection{SQM - updating rules}
Stabilizer quantum mechanics is a subtheory of quantum mechanics where we only consider common eigenstates of tensors of Pauli operators, unitaries belonging to the Clifford group, and Pauli measurements \cite{GottesPhD}. We can always write a stabilizer state $\rho$ as \begin{equation}\label{stab}\rho=\frac{1}{\mathcal{N}} \rho_1\cdot \rho_2\cdot \dots \cdot \rho_N,\end{equation} where $\mathcal{N}=Tr[\rho_1\cdot \rho_2\cdot \dots \cdot \rho_N],$ $j \in \{1,\dots, N\le n\},$ $n$ is the number of qudits and \begin{equation}\label{stabTwo} \rho_j=(\mathbb{I}_d+g_j+g_j^2+\dots+g_j^{d-1}),\end{equation} where $g_j$ is a stabilizer generator, more precisely a Weyl operator: \begin{equation}\label{weyl} \hat{W}(\bold{\lambda})=\chi(pq)\hat{S}(q)\hat{B}(p),\end{equation}
where $\chi(pq)=e^{\frac{2\pi i}{d}pq},$ $q,p$ are the coordinates of the phase space point $\lambda=(q,p),$ and  $\hat{S},\hat{B}$ are respectively the shift and boost operators (generalised Pauli operators) and the arithmetics is modulo $d,$ \begin{equation} \hat{S}(q)=\sum_{q'\in\mathbb{Z}_d}\ket{q'-q}\bra{q'}
\end{equation} 
\begin{equation}
\hat{B}(p)=\sum_{q\in\mathbb{Z}_d}\chi(pq)\ket{q}\bra{q}.
\end{equation}
% Note that the notation $g_j$ is the shortcut for $\hat{W}(\bold{\Sigma}_j).$ 
When considering more than one qudit, the Weyl operator is given by the tensor product of the single Weyl operators. 
We can write the stabilizer state $\rho$ in a more compact way as \begin{equation}\label{CompactStab}\rho = \frac{1}{\mathcal{N}}\prod_{j}^{n}\sum_{i}^{d-1} g_{j}^{i}.\end{equation}
However we will mostly use the following notation in terms of stabilizer generators, \begin{equation}\rho \rightarrow \left\langle g_1,\dots, g_N\right\rangle .\end{equation}  
We now analyse the updating rules for the state $\rho$ under the stabilizer measurement $\Pi,$ \begin{equation}\Pi \rightarrow \left\langle p_1,\dots, p_M\right\rangle \end{equation} where $p_k$ is a stabilizer generator of $\Pi$ and $k \in \{1,\dots, M\le n\}.$
We analyse the updating rules first in the commuting case ($[\rho,\Pi]=0$) and then in the general case.
\begin{enumerate}

\item For \emph{non-disturbing} (commuting) measurements, the state after measurement $\rho'$ is given by adding the stabilizer generators of the measurement $\Pi$ and the state $\rho,$ unless some generators coincide. In the latter case we obviously count them only once. \begin{equation}\label{stabcomrule}\rho' \rightarrow \left\langle g_1,g_2,\dots, g_N, p_1,p_2,\dots, p_M\right\rangle, \end{equation} where we have here considered the case in which no generators coincide. This formula means that the state $\rho'$ is now \[\rho'= \frac{1}{\mathcal{N}}\prod_{j}^{N^*}\sum_{i}^{d-1} r_{j}^{i}, \] where $N^*=N+M$ and $r_j$ is a stabilizer generator of $\rho',$ \emph{i.e.} it is either a valid (commuting) generator $g_j$ or $p_j.$  In the case where \emph{e.g.} $F$ generators coincide, then $N^*=N+M-F.$

\item For \emph{disturbing} (non-commuting) measurements (the most general case) the idea is that if we remove the non-commuting factors $\rho_j$ from the state $\rho,$ \emph{i.e.} $[\rho_j,\Pi]\neq 0,$ this case reduces to the previous commuting one. We assume the state $\rho$ to have only one non-commuting factor, say $\rho_N,$ which corresponds to the stabilizer generator $g_N.$ 
The state after measurement $\rho'$ is given by removing the non-commuting generator and adding the remaining ones of the state and measurement, unless some generators coincide. In the latter case we obviously count them only once.  
\begin{equation}\label{stabnoncomrule}\rho' \rightarrow \left\langle g_1,g_2,\dots, g_{N-1}, p_1,p_2,\dots, p_M\right\rangle, \end{equation} where we have here considered the case in which no generators coincide. This formula means that the state $\rho'$ is now \[\rho'=\frac{1}{\mathcal{N}} \prod_{j}^{N^*}\sum_{i}^{d-1} r_{j}^{i}, \] where $N^*=N+M-1$ and $r_j$ is a stabilizer generator of $\rho',$ \emph{i.e.} it is either a valid (commuting) generator $g_j$ or $p_j.$  In the case where \emph{e.g.} $F$ generators coincide, then $N^*=N+M-1-F.$
\end{enumerate}
To sum up, in the commuting case we add generators of state and measurement to obtain the state after measurement. In the non-commuting case we remove the non-commuting generator of the state and add all the others as in the commuting case. This structure is perfectly analogue to Spekkens' updating rules, which are just motivated by the classical complementarity principle.

%%%%%%%%%%%%%%%%%%%%%%%%%%%%%%%%%

\subsection{Gross' Wigner functions - updating rules}

\emph{Gross theory}. In Gross' theory the Wigner function of a state $\rho$ in a point of the phase space $\lambda\in\Omega$ is given by \begin{equation}\label{Wigner}W_{\rho}(\lambda)=Tr[\hat{A}(\lambda)\rho],\end{equation} where $\hat{A}(\lambda)$ is the phase point operator associated to each point $\lambda$, \begin{equation}\hat{A}(\lambda)=\frac{1}{d^n}\sum_{\lambda'\in\Omega}\chi(\left\langle \bold{\lambda},\bold{\lambda'}\right\rangle)\hat{W}(\bold{\lambda'}),\end{equation} where $\hat{W}(\bold{\lambda})$ are the Weyl operators defined in equation \eqref{weyl}. Note that the normalisation is such that $Tr[\hat{A}(\lambda)]=1.$
We recall that a stabilizer state is a joint eigenstate of a set of commuting Weyl operators.
Two Weyl operators commute if and only if the corresponding phase-space points $\bold{a},\bold{a'}$ have vanishing symplectic inner product:
\begin{equation} [\hat{W}(\bold{a}),\hat{W}(\bold{a'})]=0 \textrm{ if and only if } \left\langle{\bold{a},\bold{a'}}\right\rangle=\bold{a}^TJ\bold{a'}=0. \end{equation}
This result derives from the product rule of Weyl operators: \[\hat{W}(\bold{a})\hat{W}(\bold{a'})=\chi(\left\langle{\bold{a},\bold{a'}}\right\rangle)\hat{W}(\bold{a}+\bold{a'}).\]
From this result, the sets of commuting Weyl operators, and, as a consequence, the stabilizer states, are parametrized by the isotropic subspace $M$ of $\Omega.$ More precisely, for each $M$ and each $\bold{w}\in \Omega$ we can define a stabilizer state (Gross construction) $\rho_{M,\bold{w}}$ as the projector onto the joint eigenspace spanned by $\{\hat{W}(\bold{a}):\bold{a}\in M\},$ where $\hat{W}(\bold{a})$ has eigenvalue $\chi(\left\langle{\bold{w},\bold{a'}}\right\rangle).$ The Wigner function associated to the state $\rho_{M,\bold{w}}$ is always positive (necessary and sufficient condition in odd dimensions) and it is of the kind \begin{equation}\label{GrossWf}W_{(m,\bold{w})}(\bold{\lambda})=\frac{1}{d^n}\delta_{M^{C}+\bold{w}}(\lambda),\end{equation} where $M^{C}$ is the symplectic complement of $M.$ %The normalisation factor is $\frac{1}{d^n}$ only when $M$ is maximally isotropic. If it is not maximally isotropic the normalisation factor is, in general, $\frac{1}{\sum_{\lambda}\delta_{M^{C}+\bold{w}}(\lambda)}.$
Moreover the transformations that preserve the positivity of the Wigner functions are the Clifford unitaries. Gross' theory of non-negative Wigner functions is a faithful way of representing SQM.  

\emph{Equivalence of ST and GT}. The Wigner function \eqref{GrossWf} has the same form of the probability distribution \eqref{distribution} associated to the epistemic state $(V,\bold{w})$ in Spekkens' theory. More precisely, they are equivalent if we assume $M=JV,$\footnote{Note that the action of $J$ is simply to map a variable into its conjugated.} indeed this transformation implies that $V^{\perp}=M^C.$  %Even if the subspace $V$ has a more intuitive meaning, the rotation by $J$ in the phase space  allows us to adopt the same symplectic structure used by Gross (2006). %\footnote{Note that if $V$ is the set of known variable, we cannot say anything about what $V^{\perp}$ represents, since, depending on the cases, $V^{\perp}$ can represent different things (it can also have no sense in terms of known variables, imagine for example the case in which the observer knows nothing about the state of a trit). See figure 1 for a better understanding.}  %This means that $M$ is the isotropic subspace David Gross refers to in his paper. 
The equivalence between Gross' theory and Spekkens theory, using the symplectic matrix $J$ as the bridge, also extends in terms of transformations and measurement statistics \cite{Spek2}. This equivalence also implies the equivalence between Spekkens' theory and SQM in odd dimensions. Therefore we can see the description based on known variables (Spekkens) and the description based on Wigner functions (Gross) as two equivalent descriptions of stabilizer quantum mechanics in odd dimensions. We will now translate the already found updating rules of ST into Gross' Wigner functions.
%Therefore we stress again that the symplectic matrix $J$ is the bridge between Spekkens and Gross description. %(note that the latter is indeed based on the symplectic geomety). 
%In conclusion, in this section we want to focus on the bridge between the generators $\Sigma_j$ of the isotropic space V and the stabilizer generators $g_j$.

\emph{Updating rules}. Let us consider a stabilizer state $\rho=\rho_1\cdot \rho_2\cdot \dots \cdot \rho_n,$ where $n$ is the number of qudits (odd \emph{prime} dimensions), and a measurement $\Pi$ on the stabilizer state $\Pi=\Pi_1\cdot \Pi_2 \dots \cdot \Pi_m,$ where, in general, $m\leq n.$ Let us assume $m=n$ in order to consider "total" measurements (not only to a part of the state).

\newtheorem{CommutingTheorem}[Theorem]{Theorem}
\begin{CommutingTheorem}\label{CommutingTheoremWF}
\emph{Commuting case}. Let us assume the state and measurement to commute, \emph{i.e.} $[\rho,\Pi]=0.$ The Wigner function of the state after measurement is \begin{equation}\label{productrule} W_{\rho'}(\bold{\lambda})=  \frac{1}{N} W_{\rho}(\bold{\lambda})R_{\Pi}(\bold{\lambda}),\end{equation} where $\bold{\lambda}\in\Omega$ and $R_{\Pi}$ denotes the Wigner function (also called response function) associated with the measurement $\Pi.$ The normalisation factor $N$ is \[N=\sum_{\lambda\in\Omega}W_{\rho}(\bold{\lambda})R_{\Pi}(\bold{\lambda}).\]
\end{CommutingTheorem}

\begin{proof}
We rewrite the formula \eqref{productrule} by replacing the Wigner functions with their definition in terms of Spekkens' subspaces, \begin{equation}\label{deltas}\delta_{\lambda,V'^{\perp}+\bold{w'}}=\delta_{\lambda,V^{\perp}+\bold{w}}\cdot\delta_{\lambda,V_{\Pi}^{\perp}+\bold{r}}.\end{equation}
The proof is straightforward. The RHS is one if and only if both the deltas are one; this means that $\lambda$ has to belong simultaneously to $V^{\perp}+\bold{w}$ and $V_{\Pi}^{\perp}+\bold{r},$ \emph{i.e.} $\lambda\in (V^{\perp}+\bold{w})\cap (V_{\Pi}^{\perp}+\bold{r}).$ If we recall equation \eqref{ultimate} (and figure \ref{Venn}), we see that \[(V^{\perp}+\bold{w})\cap(V_{\Pi}^{\perp}+\bold{r})=(V'^{\perp})+\bold{w'},\] and we can conclude that the RHS of equation \eqref{deltas} is one if and only if the LHS is one.
At this point we can insert the normalisation factors on the RHS and the LHS. These guarantee that $\sum_{\lambda\in\Omega}W_{\rho'}(\lambda)=1$ and the uniformity as expected.
\end{proof}

In the commuting case the updating rule in SQM consists of the \emph{addition} of the stabilizer generators of state and measurement (equation \eqref{stabcomrule}). In ST the updating rule consists of the \emph{intersection} of the perpendicular isotropic subspaces (equation \eqref{ultimate}). In GT addition and intersection translate into the \emph{product} of the Wigner functions (equation \eqref{productrule}). In particular this stage consists of introducing zeros to the Wigner function in correspondence of the addition of generators to the subspace of known variables $V$ (and so removing generators from the subspace $V^{\perp}$). We will call this process - where we \emph{learn} information about the state - the \emph{localization stage}.

\newtheorem{MainTheorem}[Theorem]{Theorem}
\begin{MainTheorem}\label{MainTheorem}
\emph{Non-commuting case}. Let us assume the measurement, in general, not to commute with the state, \emph{i.e.} $[\rho,\Pi]\neq 0.$
The Wigner function of the state after measurement is \begin{equation}\label{main} W_{\rho'}(\bold{\lambda})= \frac{1}{N}\sum_{\bold{t}\in V_{other}}W_{\rho}(\bold{\lambda} - \bold{t})R_{\Pi}(\bold{\lambda}),\end{equation} where $\bold{\lambda}\in \Omega,$ $V_{other}$ is the set spanned by the non-commuting generators of Spekkens' subspace $V$ associated to the state $\rho.$ % and it corresponds to the non- commuting Weyl generator $g_n,$ \emph{i.e.} such that $[\rho_n,\Pi]\neq 0.$ The response function of the measurement is denoted by $R_{\Pi}.$ The sum between vectors is intended to be modulo $d$.
The normalisation factor $N$ is \[N=\sum_{\lambda\in\Omega}\sum_{\bold{t}\in V_{other}}W_{\rho}(\bold{\lambda} - \bold{t})R_{\Pi}(\bold{\lambda}).\]
\end{MainTheorem}

Note that we could have stated the theorem in terms of stabilizer generators instead of Spekkens' generators. The former being related to the latter as follows, \begin{equation} \label{relgen} g_j=\hat{W}(J^{-1}\bold{\Sigma}_j),\end{equation} where $J$ is the usual symplectic matrix, $\Sigma_j$ are Spekkens' generators and $g_j$ the corresponding stabilizer generators. The relation \eqref{relgen} follows from the relation between ST and GT previously described, where the bridge between the two formulations is given by the matrix $J.$

\begin{proof}
In general the state after measurement in quantum mechanics (up to a normalization) is $\rho'=\Pi\rho\Pi.$ If $[\rho,\Pi]=0$ then $\rho'=\rho\Pi.$ 

In order to simplify the proof, let us assume the case of only one non-commuting generator, say $\rho_n.$ In the present case we know, from the structure of SQM and Spekkens' updating rules (adding the commuting factors between state and measurement and removing the non-commuting ones), that the state after measurement is $\rho'=\rho^*\Pi,$ where $\rho^*=\rho_1\cdot\dots\rho_{n-1}.$
This means that we can write the state after measurement as a product of two commuting terms: $\rho^*$ and $\Pi.$ Therefore we can write the Wigner function of $\rho'$ according to the product rule for the commuting case (equation \eqref{productrule}): \[W_{\rho'}(\bold{\lambda})=\frac{1}{N}W_{\rho^*}(\bold{\lambda})R_{\Pi}(\bold{\lambda}),\] where $N=\sum_{\bold{\lambda}}W_{\rho^*}(\bold{\lambda})R_{\Pi}(\bold{\lambda}).$ % Siamo sicuri che lo possiamo fare se abbiamo stato di 2 e misura di 3? C'è da riguardare la dim per commuting per vederlo.
We want now to prove that equation \eqref{main} is equal to the latter. This means we want to prove the following:
\[W_{\rho'}(\bold{\lambda})=\sum_{\bold{t}\in V_{other}}W_{\rho}(\bold{\lambda} - \bold{t})R_{\Pi}(\bold{\lambda})=W_{\rho^*}(\bold{\lambda})R_{\Pi}(\bold{\lambda}).\]
We can simplify the terms $R_{\Pi}(\bold{\lambda}),$ thus getting \begin{equation}\label{proof}\sum_{\bold{t}\in V_{other}}W_{\rho}(\bold{\lambda} - \bold{t})=W_{\rho^*}(\bold{\lambda}).\end{equation}

At this point, in order to prove the above theorem, we rewrite the formula \eqref{main} by replacing the Wigner functions with their definition, \emph{i.e.} Kronecker deltas,
\begin{equation}\label{deltinas}\sum_{t\in V_{other}}\delta_{\bold{\lambda}-\bold{t},V^{\perp}+\bold{w}}=\delta_{\lambda,V_{commute}^{\perp}+\bold{w}},\end{equation} where $V_{commute}^{\perp} = V^{\perp}\oplus V_{other}.$
Note that we have removed the response function of the measurement. This also implies that we do not have to change $\bold{w},$ because we have only modified $V^{\perp}$ into $V^{\perp}_{commute}$ and $\bold{w'}$ is not affected. 
We now want to see that the LHS of equation \eqref{deltinas} is different from zero exactly when the RHS is. The LHS is different from zero when at least one $\bold{t}\in V_{other}$ is such that $\bold{\lambda}-\bold{t}\in V^{\perp}+\bold{w}.$ The latter corresponds to $\bold{\lambda}\in V^{\perp}+\bold{w} + \bold{t}.$ This means that $\bold{\lambda}\in V^{\perp}\oplus V_{other}+\bold{w},$ \emph{i.e.} $\bold{\lambda}\in V_{commute}^{\perp}+\bold{w},$ which is precisely what makes the RHS different from zero. 

\end{proof}

In the most general non-commuting case, in addition to the localization stage, in SQM we also have to \emph{remove} the non-commuting generators from the state (equation \eqref{stabnoncomrule}). In ST this consists of the \emph{union and shifts} in the perpendicular subspace (equation \eqref{CommuteGamma}). In GT removal and union translate into the \emph{averaging out} of the Wigner function (equation \eqref{main}). In particular this stage consists of introducing ones to the Wigner function in correspondence of the removal of generators from the subspace of known variables $V$ (and so adding generators to the subspace $V^{\perp}$). We can think of this process as the one where, after having learned some information in the localization stage, we need to forget something, otherwise we would get too much information about the ontic state, which is forbidden by the classical complementarity principle. This also explains why non-commuting measurements are also called \emph{disturbing} measurements. We will call this forgetting-part of the process the \emph{randomization stage}. Finally note that the general-case formula \eqref{main} reduce to the product rule \eqref{productrule} in the commuting case.
Figure \ref{Final} summarises the updating rules in the three theories in prime dimensions.

\begin{figure*}[h!]
\centering

{\includegraphics[width=1.0\textwidth,height=.49\textheight]{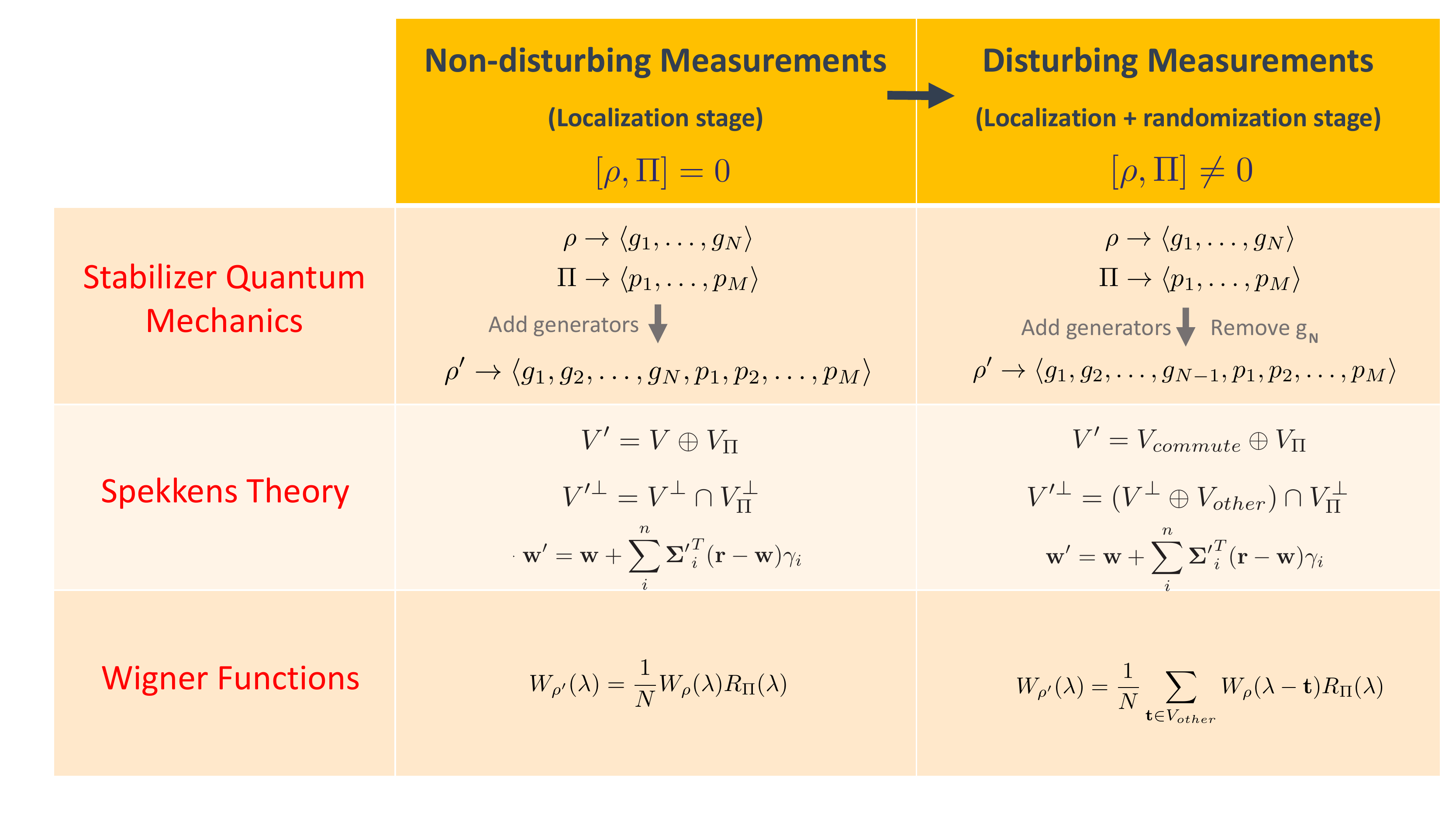}}

\caption{\footnotesize{\textbf{Equivalence of three theories in odd dimensions in terms of measurement updating rules: Spekkens' toy model, stabilizer quantum mechanics and Gross' theory.} The table above shows the updating rules in the three mentioned theories in odd prime dimensions both for the commuting and the more general non-commuting case. In SQM the updating rules were already known: if state and measurement commute then the final state $\rho'$ is given by the stabilizer generators of both $\rho$ and $\Pi.$ If, more generally, they do not commute, we also need to remove the non-commuting generators ($g_N$ in the table above) of the original state.  
In Spekkens' model the updating rules for the epistemic state $(V,\bold{w})$ and the measurement $(V_{\Pi},\bold{r})$ have the same structure of the ones in SQM. %The epistemic state is a probability distribution of the ontic states $\lambda$ on the phase space given by the kronecker delta of the subspace $V^{\perp}+\bold{w}.$ 
At the level of the perpendicular subspaces, the updating rules involve the intersection and also the direct sum (union and shifts) of the state perpendicular subspace $V^{\perp}$ with the non-commuting subspace $V_{other}$. The updating rules for the representative ontic vector $\bold{w}$ are written in terms of the measurement generators $\Sigma'_i$ and the vector $\gamma_i$ such that $\bold{\Sigma'}^{T}_i\bold{\gamma}_i=1.$
%If we combine the updating rules for the subspaces of known variables and the representative ontic vectors the updated subspaces of known variables $V'^{\perp}$ will be actually given by $V'^{\perp}=(V_{commute}^{\perp}+\bold{w}-\bold{w'})\cap(V_{\Pi}^{\perp}+\bold{r}-\bold{w'}).$
The table above does not show, for aesthetics reasons, the influence of the shift vectors $\bold{w},\bold{r},\bold{w'}$ on the perpendicular subspaces. The actual updating rule would be $V'^{\perp}=(V_{commute}^{\perp}+\bold{w}-\bold{w'})\cap(V_{\Pi}^{\perp}+\bold{r}-\bold{w'}).$
%The probability associated to each ontic state consistent with the epistemic state is uniform, $i.e.$ given by $\frac{1}{|V'^{\perp}|}.$
In Gross' theory the updating rule for Wigner functions of stabilizer states are given by a simple product of the Wigner functions associated to the state, $W_{\rho},$ and measurement, $R_{\Pi},$ in the commuting case, and an averaging over the non-commuting subspace $V_{other}$ in the general case. It is easy to see that the latter formula reduces to the previous in the commuting case (\emph{i.e.} $V_{other}=\{(0,0)\}$). %The normalisation factor $N$ is the inverse of the probability associated with each ontic state. The probability is uniform and $N$ is indeed $N=\sum_{\lambda}W_{rho}(\lambda)R_{\Pi}(\lambda)$ in the commuting case and $\sum_{\lambda}\sum_{\bold{t}\in V_{other}}W_{rho}(\lambda-\bold{t})R_{\Pi}(\lambda)$ in the general case.
}}

\label{Final}
\end{figure*}

In the non-prime dimensional case, we can rephrase all the reasonings already done in ST in terms of Wigner functions. 

\newtheorem{LemmaCoarse}[LemmaNonCommuting]{Lemma}
\begin{LemmaCoarse}\label{LemmaCoarse}
The Wigner function $W_{cg}(\lambda)$ of the coarse-graining observable $O_{cg}=a_1X_1+b_1P_1+\dots +a_nX_n+b_nP_n=D(a'_1X_1+b'_1P_1+\dots +a'_nX_n+b'_nP_n)=\sigma_{cg},$ can be written in terms of the Wigner functions $W^{(j)}_{fg}(\lambda)$ of the associated fine graining observables $O^{(j)}_{fg}=a'_1X_1+b'_1P_1+\dots +a'_nX_n+b'_nP_n=\sigma^{(j)}_{fg}$ as \begin{equation}\label{WignerNonPrime}W_{cg}(\lambda)=\frac{1}{\bar{D}}\sum_{j=0}^{\bar{D}-1}W^{(j)}_{fg}(\lambda).\end{equation}
\end{LemmaCoarse}

\begin{proof}
First of all the normalisation factor $\frac{1}{\bar{D}}$ is due to the fact that we are adding $\bar{D}$ Wigner functions, each of them having a normalisation factor of $\frac{1}{d},$ since they are Wigner functions of maximally isotropic subspaces (of dimension $d$).
The proof of the rest of the formula is straightforward. According to the definition of Wigner functions, we need to prove that \begin{equation}\label{deltaNonPrime}\delta_{V^{\perp}_{cg}+\bold{r}_{cg}}\propto\sum_{j}\delta_{V^{\perp}_{fg}+\bold{r}^{(j)}_{fg}}.\end{equation} From the decomposition of the isotropic subspaces and shift vectors in Spekkens' model, equations \eqref{FineSpek} and \eqref{ShiftFine}, we already know that $V^{\perp}_{cg}+\bold{r}_{cg}=V^{\perp}_{fg}\oplus V_D +\bold{r}_{cg} = V^{\perp}_{fg} + \sum_{j=0}^{\bar{D}-1}(\bold{r}_{cg}+j\bold{v}),$ which exactly proves that the RHS of \eqref{deltaNonPrime} is one if and only if the LHS is one.
\end{proof}

From the above construction and theorem \ref{FinalbisTheorem} we can immediately write the Wigner function of a stabilizer state after a coarse-graining measurement, thus generalising theorem \ref{MainTheorem}.

\newtheorem{FinalWignerTheorem}[Theorem]{Theorem}
\begin{FinalWignerTheorem}\label{FinalWignerTheorem}
The Wigner function of the state $\rho$ of $n$-qudit systems, where the dimension $d$ is a non-prime intger, after the (non-commuting) measurement $\Pi$ is given by
\begin{equation}\label{mainNonPrime} W_{\rho'}(\bold{\lambda})= \frac{1}{N}\frac{1}{\bar{D}}\sum_{\bold{t}\in V_{other}}\sum^{\bar{D}-1}_{j=0}W_{\rho}(\bold{\lambda} - \bold{t})R^{(j)}_{fg}(\bold{\lambda}),\end{equation} where $\bold{\lambda}\in \Omega,$ $V_{other}$ is the set spanned by the non-commuting generators of Spekkens' subspace $V$ associated to the state $\rho.$ The response function of the $j-th$ fine-graining measurement is denoted by $R^{(j)}_{fg}.$ %The sum between vectors is intended to be modulo $d$.
The normalisation factor $N$ is \[N=\sum_{\lambda\in\Omega}\sum_{\bold{t}\in V_{other}}W_{\rho}(\bold{\lambda} - \bold{t})R_{\Pi}(\bold{\lambda}),\] where $R_{\Pi}(\bold{\lambda})=\frac{1}{\bar{D}}\sum^{\bar{D}-1}_{j=0}R^{(j)}_{fg}(\bold{\lambda}).$
\end{FinalWignerTheorem}

\begin{proof}
We just need to apply lemma \ref{LemmaCoarse} to the response function of the coarse graining measurement of theorem \ref{MainTheorem}.
\end{proof}

Figure \ref{Finalbis} summarises the updating rules in ST and Gross' theory in prime and non-prime dimensions.

\begin{figure*}[h!]
\centering

{\includegraphics[width=1.0\textwidth,height=.49\textheight]{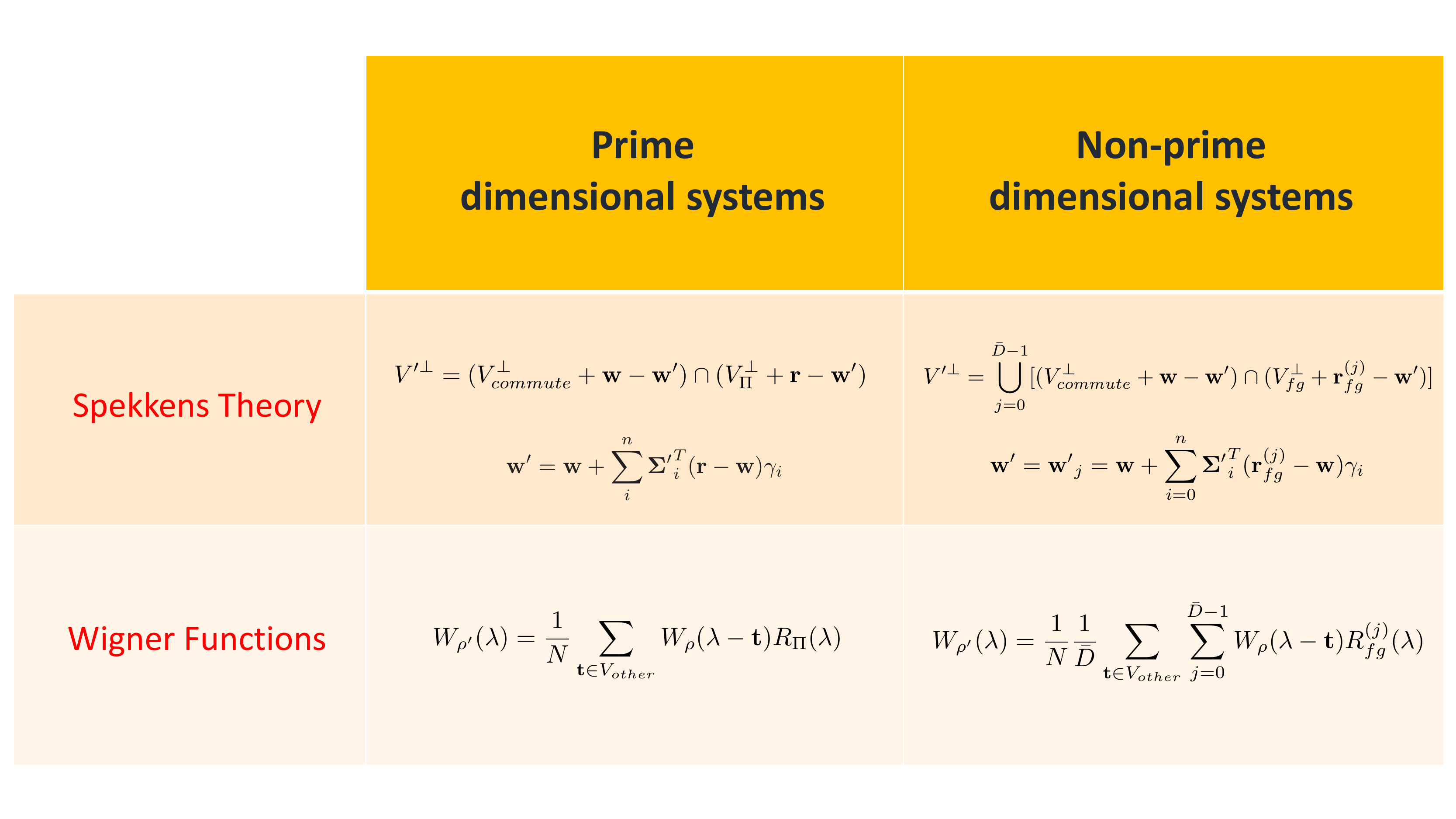}}

\caption{\footnotesize{\textbf{Measurement updating rules in Spekkens' toy model and Gross' theory in prime and non-prime dimensions.} The table above shows the updating rules (for the general non-commuting case) of ST and Gross' theory in prime dimensions, first column, and non-prime dimensions, second column. The former have been already depicted in table \ref{Final}. The latter regard the case of a coarse-graining measurement observable $O_{cg}$. In terms of perpendicular subspaces the updating rules consist of the union of the updating subspaces of the original state $(V,\bold{w})$ with each of the $\bar{D}$ individual fine-graining observables $(V_{fg},\bold{r}^{(j)}_{fg})$. The updated shift vector $\bold{w'}$ is just one of the updated shift vectors $\bold{w}'_{j}$ of the state with the fine-graining observables. In terms of Wigner functions, the union translates into a sum of $\bar{D}$ terms, and the response functions of the fine-graining observables are denoted as $R^{(j)}_{fg}.$}}

\label{Finalbis}
\end{figure*}

\section{Discussion}

The importance of completing Spekkens' theory with updating rules to determine the state after a sharp measurement relies on the possible applications of this theory for future works. In particular we think it is interesting to characterise ST in terms of its computational power, \emph{i.e.} exploring which are the quantum computational schemes that can be represented by ST. As an example ST can be used as a non-contextual hidden variable model to represent the \emph{classically simulable} part of some state-injection schemes, thus witnessing non-contextuality and also contributing to the aim of proving that contextuality is necessary for quantum speed-up in such schemes. % Citare/advertise our next work?

The result about the equivalence between ST and SQM and the associated updating rules in prime and \emph{non-prime} odd dimensions can provide a powerful new way to use and analyse SQM in non-prime dimensions, about which almost nothing is known. For example we are now facilitated to state, given a set of commuting Pauli operators, whether the joint eigenstate that they represent is pure. In non-prime dimensions the latter issue is not trivial because for coarse-graining observables the number of independent generators is not equal to the number of observables. However, from our construction to decompose coarse-graining into fine-graining observables, we know that the number of independent generators is equal to the number of fine-graining observables. Therefore if the set of commuting Pauli operators has the number of independent generators that equals the number of fine-graining observables, then the state is pure. Indeed fine-graining observables are associated to pure states. % Is it true?? fg =maximally isotropic = pure.
In addition, in the field of quantum error correction it could be interesting to study if the coarse-graining observables have any usefulness. 

Finally, the enforced equivalence of SQM, ST and Gross' theory in odd dimensions can be exploited to address a given problem from different perspectives, where, depending on the cases, one theory can be more appropriate than another. An example is the already mentioned one of addressing protocols based on SQM with Spekkens theory instead of SQM or Wigner functions.

\section{Conclusion}

Spekkens' toy model is a very powerful model which has led to meaningful insights in the field of quantum foundations and that seems to have interesting applications in the field of quantum computation. We have extended it from prime to arbitrary dimensional systems and we have derived measurement updating rules for systems of prime dimensions when the state and measurement commute, equations \eqref{ultimate}\eqref{represontic}, when they do not, equations \eqref{SubNonCom}\eqref{represontic}, and for systems of non-prime dimensions (theorem \ref{FinalbisTheorem}). These results directly derive from the basic axiom of the theory: the classical complementarity principle. The latter characterises a structure for the updating rules which is the same as in stabilizer quantum mechanics: the state after measurement is composed by the generators of the measurement and the compatible (\emph{i.e.} commuting) generators of the original state.

Spekkens showed the equivalence between SQM and ST in odd prime dimensions via Gross' Wigner functions. We have extended this result to all odd dimensions and we have translated the updating rules of ST in terms of Wigner functions (theorems \ref{CommutingTheoremWF}, \ref{MainTheorem}, \ref{FinalWignerTheorem}).
We stress again that Spekkens' model and our measurement updating rules hold in all dimensions, in even dimensions too. However the equivalence between ST and SQM only holds in odd dimensions. The main reason is that SQM in even dimensions shows contextuality, while ST does not. One of the main future challenges is to find a hidden variable toy model which is also equivalent to qubit SQM.

We treat the problem with systems of non-prime dimensions, which arises from the problem of defining an inverse in $\mathbb{Z}_d,$ by decomposing the problematic (coarse-graining) observables in terms of the non-problematic (fine-graining) ones. This approach naturally suggests the form of the updating rules.
By comparing the updating rules in the three mentioned theories we highlight the beauty and the elegance of this equivalence, where addition and removal of generators in SQM correspond to intersection and union in ST and product and randomization in GT. This correspondence is schematically depicted, for the prime-dimensional case, in table \ref{Final}. The non-prime case correspondence is represented in table \ref{Finalbis}. We believe that the fresh perspective gained by moving from one theory to another can give powerful new tools for new insights in the field of quantum computation.

\section{Acknowledgments}
We would like to thank Misja Steinmetz for suggesting Bezout's identity. This work was supported by EPSRC Centre for Doctoral Training in Delivering Quantum Technologies [EP/L015242/1].

\end{document}